\documentclass{article}
\usepackage{ijcai26}
\usepackage{algorithm}
\usepackage{algpseudocode}
\usepackage{amsfonts}
\usepackage{amsmath}
\usepackage{amsthm}
\usepackage{array}
\usepackage{bm}
\usepackage{booktabs}
\usepackage{bxtexlogo}
\usepackage[small]{caption}
\usepackage{placeins}
\usepackage{enumitem}
\usepackage{float}
\usepackage{graphicx}
\graphicspath{{figs/}}
\usepackage[hidelinks]{hyperref}
\usepackage[utf8]{inputenc}
\usepackage[switch]{lineno}
\usepackage{makecell}
\usepackage{mathtools}
\usepackage{multirow}
\usepackage{soul}
\usepackage{stfloats}
\usepackage{subcaption}
\usepackage{times}
\usepackage{url}
\usepackage[table]{xcolor}
\definecolor{bestbg}{RGB}{232,242,252}
\newcommand{\bestriskcell}[2]{%
  \cellcolor{bestbg}\riskcell{#1}{#2}%
}
\def\innerprod<#1>{\langle #1 \rangle}

\newtheorem{theorem}{Theorem}

\newtheorem{proposition}{Proposition}

\theoremstyle{definition}

\theoremstyle{remark}

\title{Rethinking Synthetic Scenario Realism:\\ Compatibility, Not Fidelity, Drives Hedging Performance}

\author{%Submission xxxx
Ryuji Hashimoto$^{1,2}$\and
Ryota Ozaki$^2$\and
Kentaro Imajo$^2$\And
Masanori Hirano$^2$\\
\affiliations
$^1$The University of Tokyo\\
$^2$Preferred Networks, Inc.
\emails
hashimoto-ryuji419@g.ecc.u-tokyo.ac.jp,
\{ryota55ozaki, imos\}@preferred.jp
research@mhirano.jp
}

\begin{document}

\maketitle

\begin{abstract}
    % 背景・問題設定・リサーチギャップ
    Deep hedging has emerged as a data-driven approach to optimizing hedging strategies. Since real market data are often limited, synthetic price paths play a crucial role in training these strategies. However, existing approaches primarily evaluate synthetic data generators based on statistical realism, while its relationship with hedging performance remains unclear.
    % 本研究のアプローチ．理論的根拠 -> 適合性の直接的検証 -> 実市場での一般的含意
    In this work, we introduce a decision-centric perspective on synthetic data for deep hedging through the notion of compatibility, which captures discrepancies in hedging risk between synthetic and real market environments. Theoretically, we establish an upper bound on hedging performance degradation in terms of learning error and a compatibility gap, and show that statistical realism and compatibility can structurally diverge. In controlled experiments, we demonstrate that the compatibility gap captures variations in hedging degradation across tasks even when generator realism is held fixed. Empirically, using real market data, we show that statistical realism does not consistently predict hedging performance, and that generator rankings depend on both hedging strategy classes and task structure.
    % Together, these findings establish compatibility as a decision-centric criterion for evaluating synthetic financial data beyond statistical realism.
\end{abstract}

\section{Introduction}
% (general motivation) ヘッジは金融機関にとって重要
Hedging is a fundamental task for financial institutions, playing a central role in risk management. In derivatives markets, where exposures evolve dynamically over time, constructing effective hedging strategies constitutes a complex sequential decision-making problem under uncertainty.

% (specific motivation) deep hedging．学習では，実データは少ないのでブラウン運動などのシナリオ生成器で代替する
In recent years, deep hedging~\cite{deep_hedging}, a data-driven approach that leverages deep learning for hedging strategies, has gained attention. In this framework, neural networks are used to directly optimize hedging strategies, enabling flexible treatment of path-dependent price dynamics, complex payoff structures, and realistic trading constraints that are often analytically intractable in traditional frameworks. In deep hedging, the limited availability of real market data typically leads to a reliance on synthetic scenario generators, which simulate synthetic underlying price paths for training~\cite{deep_hedging_w_jump}.

% (research gap) しかし，生成器の設計原理は定まっていない．既存研究ではrealism-centricな評価が一般的 (現実的な系列であるほど良い，という暗黙の仮定) だが，シナリオ生成器がrealistic -> hedgerがbetterという非自明な因果関係は検証されていない．hedgerの解釈性の重要性が指摘される中，シナリオ生成器の設計手法の知見も必要
Despite its central role, the design principles of scenario generators remain poorly understood. In the literature on synthetic data in finance, evaluation is predominantly based on the extent to which generated price paths reproduce statistical properties of real markets, i.e., realism~\cite{realism_based_evaluation,get_real}. Underlying this practice is the implicit assumption that more realistic scenarios necessarily lead to better hedging performance in real markets. In high-stakes financial settings, however, a critical question is whether hedging strategies trained on synthetic data remain effective after deployment in real markets. Yet, whether greater statistical realism actually translates into better real-world hedging performance remains an open question.

\begin{figure*}[t]
    \centering
    \includegraphics[width=\textwidth]{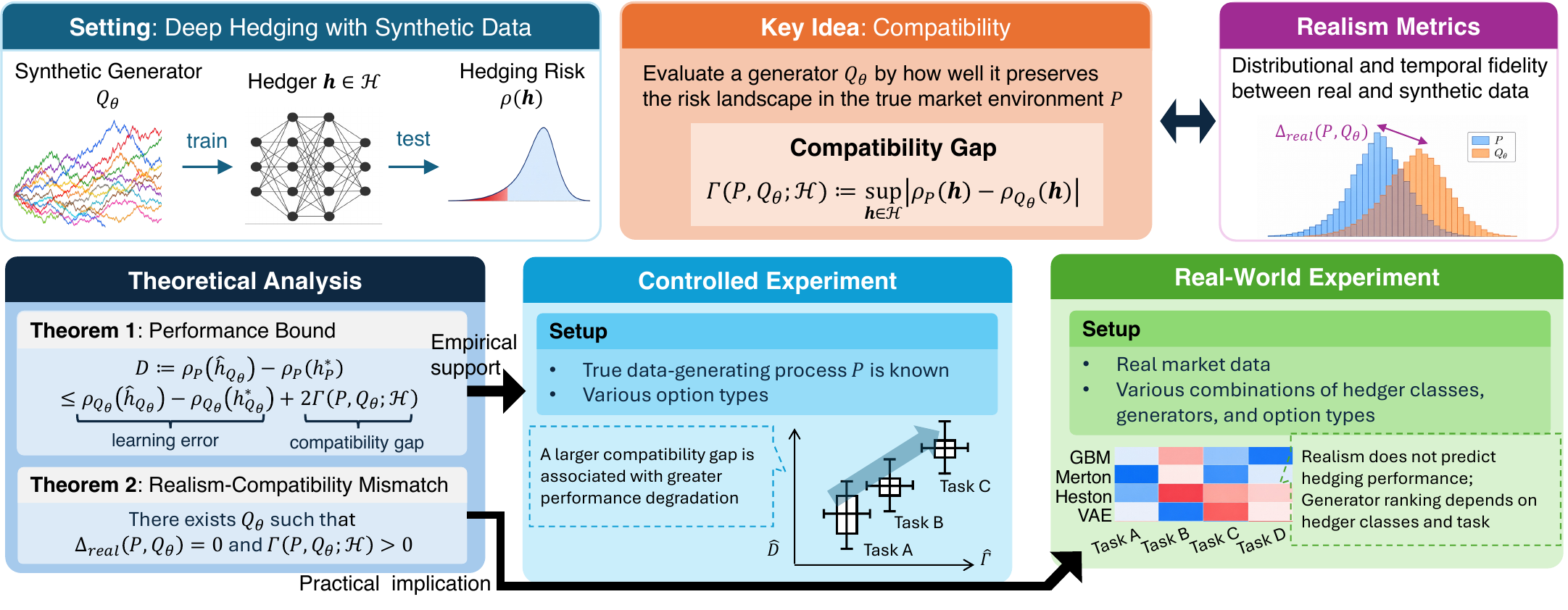}
    \caption{
        Overview of our study and key findings. We study deep hedging strategies trained on synthetic price paths and evaluated in real markets. We introduce the compatibility gap to quantify how well a synthetic generator preserves the hedging risk landscape, in contrast to conventional realism metrics. Theoretical analysis bounds hedging performance degradation in terms of learning error and compatibility gap, and establishes that realism does not guarantee compatibility. Controlled experiments demonstrate an association between compatibility gap and performance degradation, while real-world experiments show that generator rankings depend on hedger classes and hedging tasks.
    }
    \label{Fig:overview}
\end{figure*}

We investigate why distributional fidelity alone is insufficient to evaluate synthetic price generators for deep hedging. To this end, we introduce the notion of \emph{compatibility}, which captures how well a generator preserves the risk landscape relevant to a given hedging task. Specifically, we define the \emph{compatibility gap} as the discrepancy in hedging risk between the true and synthetic market environments over a given class of hedging strategies. This framework allows us to examine how generator suitability depends on the hedger class and task structure, rather than on distributional fidelity alone.

We substantiate this perspective through three complementary contributions, as summarized in Figure~\ref{Fig:overview}:
\begin{itemize}[leftmargin=*]
\item \textbf{Theoretical}: We establish an upper bound on hedging performance degradation in terms of learning error and the compatibility gap. We further show that statistical realism and compatibility can structurally diverge.
\item \textbf{Controlled}: We examine whether the compatibility gap is informative about actual hedging performance degradation beyond its role in the theoretical bound. By fixing the true and synthetic market distributions while varying hedging tasks, we find a positive association between the compatibility gap and performance degradation, despite unchanged statistical realism.
\item \textbf{Empirical}: We evaluate multiple synthetic scenario generators across hedger classes and tasks using real market data. We demonstrate that higher statistical realism does not consistently translate into better hedging performance, and that generator rankings depend on both the hedger class and task configuration.
\end{itemize}
Together, these contributions motivate a shift from realism-centered evaluation toward task-aware evaluation of synthetic scenarios in deep hedging. The suitability of a scenario generator is not an intrinsic property of its fidelity, but depends on its interaction with the hedger class and hedging task.

\section{Related Work}
% deep hedgingの二種類の方向性: hedger機構，シナリオ生成器
Research on deep hedging~\cite{deep_hedging} can be broadly categorized into two directions: the design of hedger learning mechanisms and that of scenario generators. On the hedger side, prior work has explored a wide range of learning mechanisms, including reinforcement learning~\cite{deep_rl_hedging1,deep_rl_hedging2,deep_rl_hedging3}, adversarial training frameworks that avoid explicit price process modeling~\cite{adversarial_hedging}, transformer-based architectures for capturing path dependencies~\cite{sig_former}, efficient learning frameworks~\cite{ntb_net,fast_deep_hedging}, and interpretable hedging based on market-prototype~\cite{proto_hedging}. On the scenario generator side, prior work has explored extensions of classical stochastic models, such as rough volatility models~\cite{deep_hedging_w_jump} and agent-based models~\cite{deeper_hedging}, as well as deep generative models, including GANs~\cite{deep_hedging_w_gan} and diffusion models~\cite{cofindiff,difflob}. These methods aim to produce realistic price dynamics for training, typically focusing on the faithful reproduction of statistical properties observed in financial time series~\cite{stylized_facts}. In contrast, some approaches dispense with explicit scenario generation altogether and instead learn hedging strategies directly from historical data~\cite{empirical_deep_hedging}.

Despite these advances, the two lines of research remain largely disconnected. While the design of hedging models is driven by decision performance, scenario generators are typically evaluated based on their fidelity to real-world data. As a result, it remains unclear how the properties of a scenario generator influence the hedging performance. In this work, we bridge this gap by reformulating the design of scenario generators as a decision-centric problem.

\section{Problem Formulation}
\paragraph{Market environment.} We consider a hedging problem for a single underlying asset. Let $\Omega$ denote the set of possible price paths, and let $\mathcal{P}(\Omega)$ denote the set of probability measures on $\Omega$. A price path is denoted by $\bm{\omega}=(S_0,\ldots,S_T)\in\Omega$, where $S_t\in\mathbb{R}_+$ denotes the spot price at time $t$. Let $P \in \mathcal{P}(\Omega)$ be the true probability law of price paths. $P$ represents the real market dynamics.

\paragraph{Hedging strategies.} Let $\mathcal{Y}_t$ denote the observation space available at time $t$. The hedger observes a process $Y_t=\psi_t(\bm{\omega}),~t=0,\ldots,T$, where $\psi_t:\Omega\rightarrow\mathcal{Y}_t$ is an observation map. A hedging strategy is a sequence of decision rules $\bm{h}=(h_0,\ldots,h_{T-1})\in\mathcal{H}$, where $h_t:\mathcal{I}_{t}\rightarrow\mathbb{R}$ determines the position in the underlying asset at time $t$, and $\mathcal{H}$ denotes the set of admissible hedging strategies.

\paragraph{Hedging loss and risk functional.} Let $C(\bm{\omega},\bm{h})$ denote the terminal loss incurred by hedging strategy $\bm{h}$ along path $\bm{\omega}$. Assume that $C:\Omega\times\mathcal{H}\rightarrow\mathbb{R}$ is measurable. Let $\rho:L^0(\Omega)\times\mathcal{P}(\Omega)\rightarrow\mathbb{R}$ be a risk functional used to evaluate hedging performance. For any probability measure $\mu$ on $\Omega$, define the population risk $\rho_\mu(\bm{h})\coloneq\rho_\mu(C(\bm{\omega},\bm{h}))$, where the subscript indicates that the expectation underlying $\rho$ is taken with respect to $\mu$. The objective of the hedging problem is therefore to select the hedging strategy that minimizes the risk of its resulting loss distribution under the true market dynamics: $\min_{\bm{h}\in\mathcal{H}}\rho_P(\bm{h})$. We assume that $\mathcal{H}$ is nonempty and that $\rho_P(\bm{h})$ and $\rho_{Q_\theta}(\bm{h})$ are finite for all $\bm{h}\in\mathcal{H}$. We further assume that the population and empirical risk minimizers defined below exist and that $\Gamma(P,Q_\theta;\mathcal{H})<\infty$.

\begin{figure}[t]
    \centering
    \includegraphics[width=0.49\textwidth]{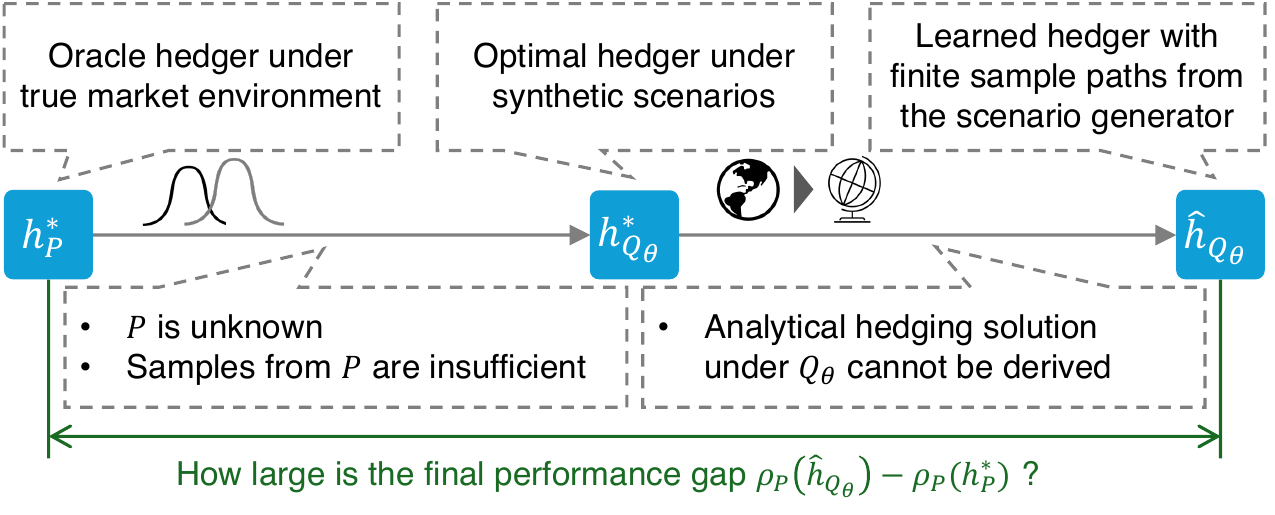}
    \caption{Two-layered approximation in deep hedging. The optimal hedger under the true distribution $P$, $\bm{h}_P^*$, is approximated via (i) replacing $P$ with the generator-induced $Q_\theta$ to obtain $\bm{h}_{Q_\theta}^*$, and (ii) learning $\hat{\bm{h}}_{Q_\theta}$ from finite samples.}
    \label{Fig:two_layered_approx}
\end{figure}

\paragraph{Training hedgers with scenario generator.} Define the optimal hedger under the true environment:
\begin{align}
\bm{h}_P^*\coloneq\arg\min_{\bm{h}\in\mathcal{H}}\rho_P(\bm{h}).\label{Eq:optimal_hedger}
\end{align}
$\bm{h}_P^*$ represents the best available strategy within $\mathcal{H}$ under the true market dynamics. However, in practice, the true process $P$ is unknown. Instead, a scenario generator is used to produce training paths. Let $Q_\theta\in\mathcal{P}(\Omega)$ denote the price-path distribution induced by a generator parametrized by $\theta$. Training then minimizes the risk under $Q_\theta$:
\begin{align}
\bm{h}_{Q_\theta}^*\coloneq\arg\min_{\bm{h}\in\mathcal{H}}\rho_{Q_\theta}(\bm{h}).
\end{align}
Our goal is therefore to obtain a hedger $\bm{h}_{Q_\theta}^*$ that achieves performance close to the optimal hedger $\bm{h}_P^*$. This requires that training under generator-induced distribution $Q_\theta$ provide a risk landscape similar to that under $P$ over the hedger class $\mathcal{H}$. We refer to this property as \textit{generator--hedger compatibility}, and define it as
\begin{align}
\Gamma(P,Q_\theta;\mathcal{H})\coloneq\sup_{\bm{h}\in\mathcal{H}}\left|\rho_P(\bm{h})-\rho_{Q_\theta}(\bm{h})\right|.
\end{align}
Moreover, deriving the optimal hedging strategy under $Q_\theta$ can become computationally challenging in high-dimensional settings with realistic trading constraints and path-dependent features.
 Consequently, in practice, training the hedger under $Q_\theta$ produces an estimator $\hat{\bm{h}}_{Q_\theta}$ obtained from a finite sample of training paths. Let $\hat{\rho}_{Q_\theta}(\bm{h})$ denote the empirical estimate of the risk, computed from a finite sample of paths $\{\bm{\omega}_i\}_{i=1}^n \sim Q_{\theta}$. In general, this corresponds to evaluating the risk functional on the empirical measure. The hedger actually obtained in practice is given by
\begin{align}
\hat{\bm{h}}_{Q_\theta}\coloneq\arg\min_{h\in\mathcal{H}}\hat{\rho}_{Q_\theta}(\bm{h}).
\end{align}
This formulation highlights that hedging entails two layers of approximation: (i) a distributional approximation, where the true law $P$ is approximated by the generator-induced law $Q_\theta$, and (ii) an optimization approximation, where the population minimizer $\bm{h}_{Q_\theta}^*$ is approximated by the empirical minimizer $\hat{\bm{h}}_{Q_\theta}$. As shown in Figure~\ref{Fig:two_layered_approx}, the final performance gap is measured under the true distribution as $\rho_P(\hat{\bm{h}}_{Q_\theta})-\rho_P(\bm{h}_P^*)$.

\section{Theoretical Insights for Generator Design}
\subsection{Decomposition of Hedger's Failure}
To understand how the use of a scenario generator affects the performance of the learned hedger, we decompose the excess risk under the true distribution $P$ into contributions from empirical optimization and distributional mismatch.

\begin{theorem}[Decomposition of the true performance gap]\label{Theorem:failure_decomp}
Let $\bm{h}_P^*$, $\bm{h}_{Q_\theta}^*$, and $\hat{\bm{h}}_{Q_\theta}$ be defined as above. Then, the excess risk under the true distribution $P$ satisfies
\begin{align}
\begin{split}
&~D(P,Q;\mathcal{H}):=\rho_P(\hat{\bm{h}}_{Q_\theta})-\rho_P(\bm{h}_P^*)\\
&\leq\underbrace{\bigl(\rho_{Q_\theta}(\hat{\bm{h}}_{Q_\theta})-\rho_{Q_\theta}(\bm{h}_{Q_\theta}^*)\bigr)}_{\text{learning error}}+\underbrace{2\Gamma(P,{Q_\theta};\mathcal{H})}_{\text{compatibility gap}}.\label{Eq:decomp_performance_gap}
\end{split}
\end{align}
\end{theorem}

\begin{proof}
By the definition of compatibility, for every $\bm{h}\in\mathcal{H}$,
\begin{align}
\left|\rho_P(\bm{h})-\rho_{Q_\theta}(\bm{h})\right|\leq\Gamma(P,Q_\theta;\mathcal{H}).
\end{align}
In particular,
\begin{align}
\rho_P(\hat{\bm{h}}_{Q_\theta})\leq\rho_{Q_\theta}(\hat{\bm{h}}_{Q_\theta})+\Gamma(P,Q_\theta;\mathcal{H}).
\end{align}
Therefore,
\begin{align}
\begin{split}
D(P,Q;\mathcal{H})&\leq\rho_{Q_\theta}(\hat{\bm{h}}_{Q_\theta})-\rho_P(\bm{h}_P^*)+\Gamma(P,Q_\theta;\mathcal{H})\\
&=\rho_{Q_\theta}(\hat{\bm{h}}_{Q_\theta})-\rho_{Q_\theta}(\bm{h}_{Q_\theta}^*)+\rho_{Q_\theta}(\bm{h}_{Q_\theta}^*)\\
&~-\rho_P(\bm{h}_P^*)+\Gamma(P,Q_\theta;\mathcal{H}).
\end{split}
\end{align}
By the optimality of $\bm{h}_{Q_\theta}^*$ under $Q_\theta$,
\begin{align}
\rho_{Q_\theta}(\bm{h}_{Q_\theta}^*)\leq\rho_{Q_\theta}(\bm{h}_P^*).
\end{align}
Hence,
\begin{align}
\begin{split}
D(P,Q;\mathcal{H})
&\leq\rho_{Q_\theta}(\hat{\bm{h}}_{Q_\theta})-\rho_{Q_\theta}(\bm{h}_{Q_\theta}^*)\\
&~+\rho_{Q_\theta}(\bm{h}_P^*)-\rho_P(\bm{h}_P^*)+\Gamma(P,Q_\theta;\mathcal{H}).
\end{split}
\end{align}
Again, by the definition of compatibility,
\begin{align}
\rho_{Q_\theta}(\bm{h}_P^*)-\rho_P(\bm{h}_P^*)\leq\left|\rho_{Q_\theta}(\bm{h}_P^*)-\rho_P(\bm{h}_P^*)\right|\leq\Gamma(P,Q_\theta;\mathcal{H}).
\end{align}
Combining the above inequalities yields
\begin{align}
D(P,Q;\mathcal{H})\leq
\rho_{Q_\theta}(\hat{\bm{h}}_{Q_\theta})-\rho_{Q_\theta}(\bm{h}_{Q_\theta}^*)+2\Gamma(P,Q_\theta;\mathcal{H}),
\end{align}
which proves Eq.~\eqref{Eq:decomp_performance_gap}.
\end{proof}

Theorem~\ref{Theorem:failure_decomp} shows that the excess hedging risk is upper-bounded by two terms: the learning error arising from finite data and optimization, and the compatibility gap induced by replacing $P$ with $Q_\theta$. This decomposition implies that a generator need not approximate $P$ globally; rather, it suffices that $Q_\theta$ preserves the risk evaluation of strategies in $\mathcal{H}$ provided that the learning error is also controlled. In other words, what matters is not global realism, but \emph{decision-relevant realism}, i.e., the extent to which $Q_\theta$ induces a risk landscape aligned with that of $P$ over $\mathcal{H}$.

Moreover, Theorem~\ref{Theorem:failure_decomp} suggests a more subtle trade-off: even if using $Q_\theta$ introduces a compatibility gap, it may still lead to superior final performance if the induced simplification reduces the learning error sufficiently. This highlights that an appropriate generator can outperform a more realistic one when it better balances compatibility and learnability for the given hedger class.

\subsection{Mismatch of Realism and Compatibility}

A central question regarding Theorem~\ref{Theorem:failure_decomp} is whether improving the realism of $Q_\theta$ necessarily leads to better hedging performance. By conceptually fixing the learnability of $Q_\theta$ for the hedger, we examine whether reducing the realism discrepancy between $P$ and $Q_\theta$ necessarily reduces the compatibility gap $\Gamma(P,Q_\theta;\mathcal{H})$. This asks whether realism and compatibility are aligned notions. Realism is a property of the generator $Q_\theta$ relative to the true process $P$, typically defined in terms of distributional proximity. In contrast, compatibility is defined relative to a hedger class $\mathcal{H}$, and measures discrepancies between $P$ and $Q_\theta$ relevant to $\mathcal{H}$. In other words, compatibility depends on how the difference between $P$ and $Q_\theta$ is \textit{seen} through the loss functions induced by $\mathcal{H}$.

Below, we first establish a sufficient condition under which realism upper-bounds compatibility, so that a small realism discrepancy guarantees a small compatibility gap. We then show that, without this condition, the two notions can diverge: there exist generators that perfectly match the evaluated realism features yet remain incompatible.

We specialize the risk functional to the expectation case:
$\rho_\mu(\bm{h})\coloneq
\mathbb{E}_{\bm{\omega}\sim\mu}
[C(\bm{\omega},\bm{h})]$.
Accordingly, the compatibility gap becomes
\begin{align}
\Gamma(P,Q_\theta;\mathcal{H})
=\sup_{\bm{h}\in\mathcal{H}}
\left|
\mathbb{E}_P[C(\bm{\omega},\bm{h})]
-\mathbb{E}_{Q_\theta}[C(\bm{\omega},\bm{h})]
\right|.
\end{align}

\paragraph{Realism and compatibility as integral probability metrics.}
We express both realism and compatibility within a unified framework based on integral probability metrics (IPMs). Let $\mathcal{F}_{\mathrm{real}}$ be a uniformly bounded class of measurable functions $f:\Omega\rightarrow\mathbb{R}$. The associated realism discrepancy between $P$ and $Q_\theta$ is defined by
\begin{align}
\Delta_{\mathrm{real}}
(P,Q_\theta;\mathcal{F}_{\mathrm{real}})
\coloneq
\sup_{f\in\mathcal{F}_{\mathrm{real}}}
\left|
\mathbb{E}_P[f(\bm{\omega})]
-\mathbb{E}_{Q_\theta}[f(\bm{\omega})]
\right|.
\label{Eq:Delta_real}
\end{align}
Typically, each function $f$ extracts a particular feature or statistic of the outcome $\bm{\omega}$, and $\Delta_{\mathrm{real}}$ measures the largest discrepancy over the feature class $\mathcal{F}_{\mathrm{real}}$. Similarly, define the hedging loss class:
\begin{align}
\mathcal{F}_{\mathcal{H}}
\coloneq
\left\{
f_{\bm{h}}:\Omega\rightarrow\mathbb{R}
\mid
f_{\bm{h}}(\bm{\omega})=C(\bm{\omega},\bm{h}),
\bm{h}\in\mathcal{H}
\right\}.
\end{align}
Each function $f_{\bm{h}}$ corresponds to the loss induced by a particular hedging strategy $\bm{h}$, mapping each price path $\bm{\omega}$ to its resulting hedging cost. Then the compatibility gap can be written as the IPM over $\mathcal{F}_{\mathcal{H}}$:
\begin{align}
\Gamma(P,Q_\theta;\mathcal{H})
=
\sup_{f\in\mathcal{F}_{\mathcal{H}}}
\left|
\mathbb{E}_P[f(\bm{\omega})]
-\mathbb{E}_{Q_\theta}[f(\bm{\omega})]
\right|.
\label{Eq:Gamma_ipm}
\end{align}
Thus, compatibility is also an IPM, but over the decision-relevant loss function class $\mathcal{F}_{\mathcal{H}}$, rather than the test functions capturing distributional features used in realism evaluation $\mathcal{F}_{\mathrm{real}}$. This representation assumes that the hedging losses are integrable under the distributions considered.

\begin{proposition}[A sufficient condition for realism-compatibility alignment]
\label{Prop:sufficient_alignment_condition}
If $\mathcal{F}_{\mathcal{H}}
\subseteq\mathcal{F}_{\mathrm{real}}$, then
\begin{align}
\forall P,Q_\theta,\quad
\Gamma(P,Q_\theta;\mathcal{H})
\leq
\Delta_{\mathrm{real}}
(P,Q_\theta;\mathcal{F}_{\mathrm{real}}).
\end{align}
\end{proposition}

\begin{proof}
Since $\mathcal{F}_{\mathcal{H}}
\subseteq\mathcal{F}_{\mathrm{real}}$,
\begin{align}
\sup_{f\in\mathcal{F}_{\mathcal{H}}}\left|\mathbb{E}_P[f]-\mathbb{E}_{Q_\theta}[f]\right|\leq\sup_{f\in\mathcal{F}_{\mathrm{real}}}\left|\mathbb{E}_P[f]-\mathbb{E}_{Q_\theta}[f]\right|.
\end{align}
The claim follows directly from
Eqs.~(\ref{Eq:Delta_real}) and
(\ref{Eq:Gamma_ipm}).
\end{proof}

Proposition~\ref{Prop:sufficient_alignment_condition} shows that, if the realism evaluation class includes all loss functions induced by admissible hedging strategies, then realism upper-bounds compatibility. Thus, a sufficiently small realism discrepancy guarantees a small compatibility gap, although reducing realism discrepancy does not necessarily reduce compatibility monotonically. However, the inclusion condition is strong: realism evaluations typically capture summary statistics, such as moments and tail measures, rather than the loss functions of a constrained hedger class.

\begin{theorem}[Small realism does not imply compatibility]\label{Theorem:small_realism_incompatibility}
Assume that
\begin{align}
\mathcal{F}_{\mathrm{real}}
=
\left\{
f\in\mathrm{span}\{\phi_1,\ldots,\phi_m\}
:
\|f\|_\infty\leq 1
\right\},
\end{align}
where $1,\phi_1,\ldots,\phi_m$ are linearly independent bounded measurable functions on $\Omega$, and
$\|f\|_\infty
=\sup_{\bm{\omega}\in\Omega}|f(\bm{\omega})|$.
Assume further that there exists some strategy $\bm{h}^\dagger\in\mathcal{H}$ such that
\begin{align}
f_{\bm{h}^\dagger}
\notin
\mathrm{span}\{1,\phi_1,\ldots,\phi_m\}.
\end{align}
Then, there exist probability measures
$P,Q_\theta\in\mathcal{P}(\Omega)$ such that
\begin{align}
\Delta_{\mathrm{real}}(P,Q_\theta;\mathcal{F}_{\mathrm{real}})=0,~\Gamma(P,Q_\theta;\mathcal{H})>0.
\end{align}
\end{theorem}

\begin{proof}
By assumption, $g\notin\mathrm{span}\{1,\phi_1,\ldots,\phi_m\}$, where $g\coloneq f_{\bm{h}^\dagger}$. Hence, the $m+2$ functions
$1,\phi_1,\ldots,\phi_m,g$ are linearly independent. Therefore, there exist points
$\bm{\omega}_1,\ldots,\bm{\omega}_{m+2}\in\Omega$
such that the matrix
\begin{align}
A=
\begin{pmatrix}
1 & \cdots & 1\\
\phi_1(\bm{\omega}_1) & \cdots &
\phi_1(\bm{\omega}_{m+2})\\
\vdots & \ddots & \vdots\\
\phi_m(\bm{\omega}_1) & \cdots &
\phi_m(\bm{\omega}_{m+2})\\
g(\bm{\omega}_1) & \cdots &
g(\bm{\omega}_{m+2})
\end{pmatrix}
\end{align}
has full rank $m+2$.

Let $B\in\mathbb{R}^{(m+1)\times(m+2)}$
be the submatrix consisting of the first $m+1$ rows of $A$. Since $B$ has more columns than rows, there exists a nonzero vector
$\bm{a}=(a_1,\ldots,a_{m+2})^\top
\in\ker(B)$.
Moreover, because $A$ has full rank,
\begin{align}
\sum_{i=1}^{m+2}
a_i g(\bm{\omega}_i)\neq 0.
\end{align}
Define the signed measure
\begin{align}
\mu\coloneq
\sum_{i=1}^{m+2}
a_i\delta_{\bm{\omega}_i}.
\end{align}
Since $\bm{a}\in\ker(B)$,
\begin{align}
\mu(\Omega)=\sum_{i=1}^{m+2}a_i=0,~\int f\,d\mu=0~\forall f\in\mathcal{F}_{\mathrm{real}},~\int g\,d\mu\neq0.
\end{align}
Write the Jordan decomposition of $\mu$ as
$\mu=\mu^+-\mu^-$. Since $\mu(\Omega)=0$,
\begin{align}
\mu^+(\Omega)=\mu^-(\Omega)
\eqcolon c>0.
\end{align}
Define probability measures
\begin{align}
P\coloneq\frac{\mu^+}{c},
\qquad
Q_\theta\coloneq\frac{\mu^-}{c}.
\end{align}
Then, for every
$f\in\mathcal{F}_{\mathrm{real}}$,
\begin{align}
\mathbb{E}_P[f]-\mathbb{E}_{Q_\theta}[f]
=\frac{1}{c}\int f\,d\mu=0.
\end{align}
Therefore,
\begin{align}
\Delta_{\mathrm{real}}
(P,Q_\theta;\mathcal{F}_{\mathrm{real}})=0.
\end{align}
On the other hand,
\begin{align}
\mathbb{E}_P[g]-\mathbb{E}_{Q_\theta}[g]
=\frac{1}{c}\int g\,d\mu\neq0.
\end{align}
Since $g=f_{\bm{h}^\dagger}$,
\begin{align}
\Gamma(P,Q_\theta;\mathcal{H})\geq\left|\mathbb{E}_P[f_{\bm{h}^\dagger}]-\mathbb{E}_{Q_\theta}[f_{\bm{h}^\dagger}]\right|>0.
\end{align}
This proves the claim.
\end{proof}

Theorem~\ref{Theorem:small_realism_incompatibility} establishes a structural distinction between realism and compatibility. Whenever the realism function class fails to span a hedging-relevant loss function, one can construct distributions that are indistinguishable under the realism metric yet induce different hedging risks. Importantly, this is a constructive result: even perfect agreement on the evaluated realism features does not guarantee compatibility. Thus, a generator may distort the risk landscape perceived by the hedger despite reproducing the selected statistical features. While a positive compatibility gap does not itself imply degradation in out-of-sample hedging performance, it shows that realism alone cannot guarantee the small compatibility gap required by Theorem~\ref{Theorem:failure_decomp}.

\section{Experiments}\label{Sec:Experiments}
We conduct two complementary experiments. First, building on Theorem~\ref{Theorem:failure_decomp}, we examine whether the compatibility gap is informative about actual hedging performance degradation beyond providing a theoretical upper bound. To this end, we conduct a controlled experiment with fixed market distributions and varying hedging tasks. Second, motivated by Theorem~\ref{Theorem:small_realism_incompatibility}, we investigate whether statistical realism reliably predicts downstream hedging performance in practice. We compare multiple generators across hedger classes and tasks using real market data.

\subsection{Common Hedging Setup}\label{Sec:common_hedging_setup}
We hedge European and Lookback call options over 20 daily trading steps, with initial price $S_0=1$. A task specifies the option type, strike price, and proportional transaction cost. Strategies are trained on synthetic price paths and evaluated using a risk measure of terminal hedging loss $C(\bm{\omega},\bm{h})$. Appendix~\ref{App:hedging_setup} specifies the payoffs of these option types and risk measures. Experiment-specific settings are reported separately below.

\subsection{Controlled Experiment}\label{Sec:controlled_experiment}
We fix the target distribution $P$ to a Heston process~\cite{heston} and the synthetic distribution $Q_\theta$ to geometric Brownian motion (GBM). Thus, their statistical realism discrepancy is unchanged across tasks. We adopt a common no-transaction band network (NTBNet) architecture~\cite{ntb_net} with bounded positions to define an economically meaningful hedger class $\mathcal H$, avoiding trivial inflation of the compatibility gap by extreme trading strategies. We vary the option type, strike $K\in\{0.98,1.00,1.02\}$, and transaction cost $\{10^{-4},10^{-3}\}$, giving $12~(=2\times 3\times2)$ tasks. All strategies are trained to minimize Conditional Value-at-Risk (CVaR) at confidence level $0.95$, which evaluates the mean loss in the worst $5\%$ of outcomes. See Appendix~\ref{App:hedging_setup} for its formal definition. The architecture and position constraints are fixed, while each task's inputs and reference delta follow its option specification. Within each task, training and compatibility estimation use the same admissible class $\mathcal H$.

We approximate the unavailable optimal strategy $\bm h_P^*$ by $\tilde{\bm h}_P$, the strategy with the lowest validation risk across 50 training trials under $P$. For each of 50 training trials under $Q_\theta$, we select the checkpoint with the lowest validation risk under $Q_\theta$ and evaluate it on test paths from $P$. We estimate performance degradation by subtracting the test risk of $\tilde{\bm h}_P$, which is shared across all 50 trials within the same task. The estimated performance degradation is
\begin{align}
\hat{D}=\hat\rho_P^{\mathrm{test}}(\hat{\bm h}_{Q_\theta})-\hat\rho_P^{\mathrm{test}}(\tilde{\bm h}_P),
\end{align}
where $\hat\rho_\mu^{\mathrm{test}}$ denotes empirical risk on independent test paths from $\mu$.

To estimate the compatibility gap, we optimize the NTBNet parameters to maximize each sign of the empirical risk difference between $P$ and $Q_\theta$. Let $\bm h_{\bm w}$ denote the hedging strategy represented by NTBNet with parameters $\bm w$. For each sign $s\in\{-1,+1\}$, we perform gradient ascent using
\begin{align}
s\left[\nabla_{\bm w}\hat\rho_P^{\mathrm{batch}}(\bm h_{\bm w})-\nabla_{\bm w}\hat\rho_{Q_\theta}^{\mathrm{batch}}(\bm h_{\bm w})\right],
\end{align}
where $\hat\rho_\mu^{\mathrm{batch}}$ denotes empirical CVaR on a minibatch from the search dataset for $\mu$. Among the candidate checkpoints from both signs and multiple initializations, we select the strategy $\tilde{\bm h}$ with the largest absolute risk difference on validation data. We then estimate the compatibility gap on independent test data as
\begin{align}
\hat{\Gamma}=\left|\hat\rho_P^{\mathrm{test}}(\tilde{\bm h})-\hat\rho_{Q_\theta}^{\mathrm{test}}(\tilde{\bm h})\right|.
\end{align}
We repeat the complete search 10 times. Appendix~\ref{App:controlled_details} reports the simulation parameters and estimation procedures.

\subsection{Real-World Experiment}\label{Sec:real_world_experiment}
We examine whether statistical realism predicts hedging performance and whether generator rankings persist across hedger classes and tasks. To this end, we prepare four scenario generators with different levels of statistical structure, four task configurations with different hedging difficulty, and three hedgers with different inductive biases and model capacities. We calibrate generators using daily Nikkei 225 (N225) data from 2018--2022 and evaluate hedgers on data from January 2023 to February 2026. The four tasks combine European and Lookback calls with transaction costs of $10^{-4}$ and $10^{-3}$, with $K=1$.

\paragraph{Scenario generators.} We compare GBM, the Merton jump process~\cite{merton}, Heston process, and a variational autoencoder (VAE). The VAE learns a latent representation of return sequences. Calibration procedures and model specifications are provided in Appendix~\ref{App:real_world_details}.

\paragraph{Hedging models.} We compare three stationary hedger classes. LinearHedger determines the position as a linear combination of input features. MLPHedger directly predicts the position using a multilayer perceptron. NTBNet predicts a no-transaction band and projects the previous position onto that band when it lies outside. LinearHedger and MLPHedger use log-moneyness, time-to-expiry, realized volatility over the previous 10 days, and the previous position. NTBNet uses the Black--Scholes delta and the previous position. Network architectures are specified in Appendix~\ref{App:real_world_details}.

\paragraph{Risk measures and hedger performance.} Each configuration is trained and evaluated separately using the entropic risk measure (ERM) with risk-aversion parameter $50$, or CVaR at confidence level $0.90$ or $0.95$. ERM penalizes large losses through exponential weighting; CVaR evaluates the mean loss in the worst $10\%$ or $5\%$ of outcomes, respectively. Their formal definitions are given in Appendix~\ref{App:hedging_setup}. Each strategy is evaluated using the same risk measure and parameter as its training objective. For each hedger--task setting, we rank the four generators by test risk within each risk measure and report their mean rank.

\paragraph{Statistical realism.}
We compare synthetic and real daily log returns using eight distances, with smaller values indicating closer agreement:
\begin{itemize}[leftmargin=*]
\item \textbf{KS}: Kolmogorov--Smirnov distance between return distributions.
\item \textbf{MeanAbsRet}: absolute difference in mean absolute return.
\item \textbf{Kurt}: absolute difference in return kurtosis.
\item \textbf{Skew}: absolute difference in return skewness.
\item \textbf{Hill (Avg)}: mean absolute difference in Hill tail-index estimates~\cite{hill}, which characterize tail heaviness, for positive, negative, and absolute returns, using the largest $5\%$ of tail magnitudes.
\item \textbf{ACF (Raw)}: mean absolute difference in the autocorrelation function (ACF) of returns over lags $1$--$9$.
\item \textbf{ACF (Abs)}: the corresponding ACF distance for absolute returns.
\item \textbf{Hurst (Abs)}: absolute difference in the estimated Hurst exponent of absolute returns, which characterizes temporal persistence.
\end{itemize}
Appendix~\ref{App:realism_diagnostics} provides further details. Unlike the controlled experiment, the real-world experiment does not estimate $\Gamma$ or $D$, because neither the true market distribution nor its optimal strategy is available. It instead tests whether statistical realism yields a consistent performance ranking across hedger classes and tasks.

\section{Results and Discussions}\label{Sec:ResultsAndDiscussions}

\subsection{Controlled Experiment}\label{Sec:controlled_results}

\begin{figure}[t]
    \centering
    \includegraphics[width=\columnwidth]{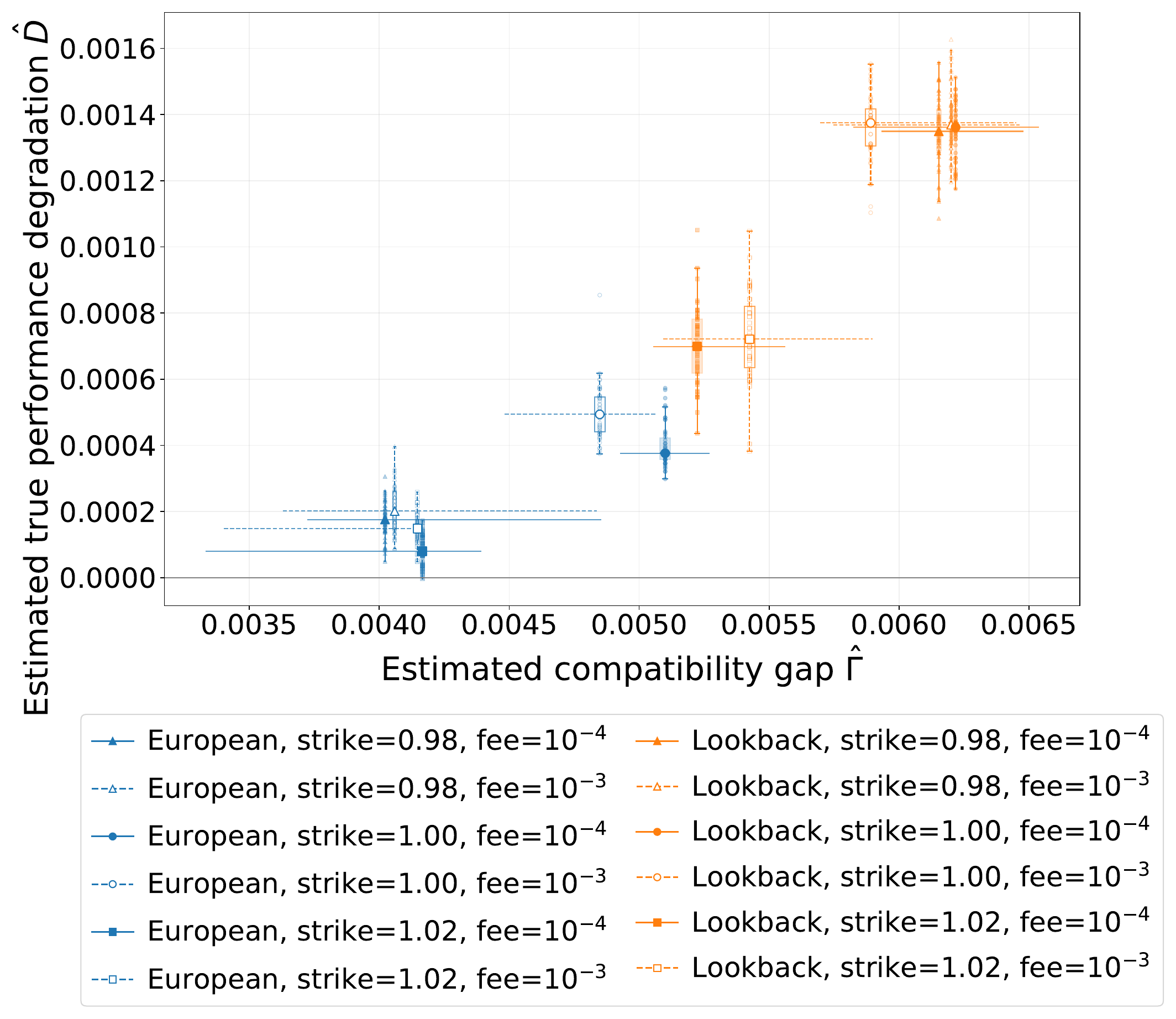}
    \caption{Controlled experiment with fixed market distributions. Horizontal coordinates show the median $\hat{\Gamma}$ from 10 searches; horizontal lines show their minimum and maximum, not confidence intervals. Vertical points show $\hat{D}$ from 50 training trials. Boxes show interquartile ranges.}
    \label{Fig:compatibility_degradation_scatter}
\end{figure}

Figure~\ref{Fig:compatibility_degradation_scatter} shows that the compatibility gap varies substantially across hedging tasks, despite the true and synthetic market distributions being fixed. Moreover, larger compatibility gaps are generally associated with greater hedging performance degradation. While Theorem~\ref{Theorem:failure_decomp} establishes an upper bound rather than a monotone relationship, this result suggests that the compatibility gap is informative about realized degradation beyond its theoretical role in the bound. Importantly, the observed variation cannot be explained by differences in statistical realism between the market distributions, which remains unchanged across tasks.

\subsection{Real-World Experiment}\label{Sec:real_world_results}

\begin{table}[t]
    \centering
    \caption{Realism metrics for the four scenario generators (lower is better). \textbf{Bold} indicates the best value in each column.}
    \label{Tab:realism_metrics}
    \small
    \setlength{\tabcolsep}{3.5pt}
    \begin{tabular}{@{}lrrrr@{}}
        \toprule
        & GBM & Merton & Heston & VAE \\
        \midrule
        KS ($10^{-2}$) & 8.619 & 7.992 & \textbf{4.315} & 7.227 \\
        MeanAbsRet ($10^{-4}$) & 17.084 & 12.170 & 3.444 & \textbf{1.120} \\
        Kurt & 16.032 & \textbf{12.752} & 13.876 & 15.145 \\
        Skew ($10^{-1}$) & 8.619 & 7.401 & 6.527 & \textbf{6.399} \\
        Hill (Avg) & 4.047 & \textbf{0.237} & 1.963 & 2.533 \\
        ACF (Raw) ($10^{-2}$) & 4.172 & 4.250 & 5.153 & \textbf{1.773} \\
        ACF (Abs) ($10^{-1}$) & 1.875 & 1.946 & \textbf{0.809} & 1.028 \\
        Hurst (Abs) ($10^{-1}$) & 1.331 & 1.239 & 0.789 & \textbf{0.678} \\
        \bottomrule
    \end{tabular}
\end{table}

\paragraph{Realism of the generators.} Table~\ref{Tab:realism_metrics} reports eight realism metrics for the four scenario generators. Appendix~\ref{App:realism_diagnostics} provides their corresponding return distributions and autocorrelation structures, together with those of real data. GBM performs the worst across almost all metrics, reflecting its simplifying assumptions of Gaussianity and independent increments. It fails to capture both heavy tails and temporal dependencies observed in real data. Notably, VAE outperforms GBM across all eight realism metrics, indicating a consistent advantage in statistical fidelity. Merton improves tail behavior by incorporating jump components, achieving better performance in tail-related metrics (e.g., Kurt). Heston provides a more balanced representation by introducing stochastic volatility, enabling it to capture both distributional properties and temporal dynamics to a certain extent. This is reflected in its relatively strong performance in both ACF-based and distributional metrics (e.g., KS). VAE, as a data-driven generative model, shows strength in capturing long-range dependencies, leading to improved performance in metrics related to temporal structure (e.g., Hurst exponent).

\paragraph{Hedger performance across tasks.} Table~\ref{Tab:best_models} reports the best-performing generator--hedger pair for each task configuration. As shown in the table, no single scenario generator or hedger consistently achieves the top performance across all tasks, indicating that there is no globally optimal choice independent of the task setting.

\begin{table}[t]
\centering
\caption{Best generator--hedger pair for each task configuration. For each task, the pair that achieves the best average performance over ERM(50), CVaR(0.90) and CVaR(0.95) among all generator--hedger combinations is reported.}
\label{Tab:best_models}
{\small
\begin{tabular}{ll}
\toprule
Option (fee) & Best Generator--Hedger \\
\midrule
European ($10^{-4}$) & Merton--MLPHedger \\
European ($10^{-3}$) & GBM--LinearHedger \\
Lookback ($10^{-4}$) & Heston--MLPHedger \\
Lookback ($10^{-3}$) & GBM--NTBHedger \\
\bottomrule
\end{tabular}}
\end{table}

Table~\ref{Tab:generator_ranking} reports the generator rankings for each hedger--task setting,  where each task is defined by the option type and transaction cost. The tasks are ordered by increasing difficulty: European options with low and high costs, followed by path-dependent Lookback options with low and high costs. Importantly, the generator rankings do not preserve the ordering suggested by statistical realism. Despite its relatively poor realism, GBM achieves competitive hedging performance in some settings, particularly for simpler hedging strategies and less demanding tasks, ranking first under LinearHedger for European options with low transaction costs ($10^{-4}$). Conversely, VAE, which outperforms GBM across all eight realism metrics, ranks below GBM in several hedging settings.

Generator rankings also depend on the hedger class, even for the same task. For example, Heston ranks 1.67 under LinearHedger but falls to last place under NTBNet for European options with low transaction costs. Similarly, GBM drops from 1.67 to last place when switching from LinearHedger to NTBNet for European options with high transaction costs ($10^{-3}$). Furthermore, within MLPHedger, generator rankings change across option types and transaction costs, demonstrating their dependence on task characteristics.

Together, these findings are consistent with the theoretical distinction established in Theorem~\ref{Theorem:small_realism_incompatibility}: statistical realism alone does not reliably predict downstream hedging performance. Generator suitability depends on the hedger class and task, consistent with the task-dependent perspective of compatibility. Appendix~\ref{App:learnability_diagnostics} provides complementary optimization diagnostics, while the relative contributions of learning error and compatibility gap remain unidentified. Appendix~\ref{App:generator_risk_values} reports the detailed risk values underlying the generator rankings.

\begin{table}[t]
\centering
\caption{Generator ranking (lower is better) for each hedger–task setting. Each cell reports the average rank over ERM$(50)$, CVaR$(0.90)$, and CVaR$(0.95)$. \textbf{Bold} values indicate the best (lowest) rank within each hedger–task setting.}
\label{Tab:generator_ranking}
\setlength{\tabcolsep}{4pt}
\renewcommand{\arraystretch}{1.1}
{\footnotesize
\begin{tabular}{lcccc}
\toprule
\multicolumn{5}{c}{LinearHedger} \\
\cmidrule(lr){1-5}
Option (fee) & GBM & Merton & Heston & VAE \\
\midrule
European ($10^{-4}$) & \textbf{1.33} & 3.50 & 1.67 & 3.50 \\
European ($10^{-3}$) & 1.67 & 4.00 & 3.00 & \textbf{1.33} \\
Lookback ($10^{-4}$) & \textbf{1.00} & 4.00 & 2.00 & 3.00 \\
Lookback ($10^{-3}$) & 2.00 & 4.00 & 3.00 & \textbf{1.00} \\
\midrule
\multicolumn{5}{c}{NTBNet} \\
\cmidrule(lr){1-5}
Option (fee) & GBM & Merton & Heston & VAE \\
\midrule
European ($10^{-4}$) & 2.00 & 3.00 & 4.00 & \textbf{1.00} \\
European ($10^{-3}$) & 4.00 & 3.00 & \textbf{1.33} & 1.67 \\
Lookback ($10^{-4}$) & 3.00 & 2.00 & 4.00 & \textbf{1.00} \\
Lookback ($10^{-3}$) & \textbf{1.00} & 2.00 & 3.50 & 3.50 \\
\midrule
\multicolumn{5}{c}{MLPHedger} \\
\cmidrule(lr){1-5}
Option (fee) & GBM & Merton & Heston & VAE \\
\midrule
European ($10^{-4}$)
& 2.00 & \textbf{1.00} & 3.00 & 4.00 \\
European ($10^{-3}$)
& 4.00 & 2.00 & 3.00 & \textbf{1.00} \\
Lookback ($10^{-4}$)
& 4.00 & 3.00 & \textbf{1.00} & 2.00 \\
Lookback ($10^{-3}$)
& 4.00 & 3.00 & 2.00 & \textbf{1.00} \\
\bottomrule
\end{tabular}}
\end{table}

\section{Conclusion}
We reconsider the role of synthetic scenario generators in deep hedging, shifting the focus from realism-centric to decision-centric evaluation. Challenging the assumption that higher realism improves hedging performance, we show theoretically that performance degradation is bounded by learning error and a compatibility gap, and that realism and compatibility need not align. Our controlled experiment shows that the compatibility gap is associated with actual performance degradation across tasks, even under fixed market distributions. Our real-world experiment demonstrates that higher statistical realism does not consistently improve hedging performance, with generator rankings varying across hedger classes and tasks. Together, these findings suggest that the suitability of synthetic scenarios is not determined by statistical fidelity alone, but depends on their interaction with the hedger class and hedging task.

Our study has limitations. Compatibility serves as a theoretical concept rather than a directly measurable evaluation metric, as the true market distribution is unknown. Developing reliable proxies for compatibility remains an important direction for future work. Another promising direction is compatibility-aware deep hedging, jointly designing scenario generation and hedging to improve downstream decision performance. This perspective opens new possibilities for generative finance beyond fidelity-driven approaches.

% \section*{Ethical Statement}

% \section*{Acknowledgments}

%% The file named.bst is a bibliography style file for BibTeX 0.99c
\bibliographystyle{named}
\bibliography{ijcai26}

@misc{pfhedge,
  author       = {Imaki, Shota and Hirano, Masanori},
  title        = {{PFHedge}: Deep Hedging in {PyTorch}},
  howpublished = {\url{https://github.com/pfnet-research/pfhedge}},
  year         = {2024}
}

@article{deep_hedging,
    author = {Buehler, Hans and Gonon, Lukas and Teichmann, Josef and Wood, Ben},
    title = {Deep hedging},
    journal = {Quantitative Finance},
    volume = {19},
    number = {8},
    pages = {1271--1291},
    year = {2019},
    doi = {10.1080/14697688.2019.1571683}
}

@inproceedings{deep_rl_hedging1,
    author = {Murray, Phillip and Wood, Ben and Buehler, Hans and Wiese, Magnus and Pakkanen, Mikko},
    title = {{Deep hedging: Continuous reinforcement learning for hedging of general portfolios across multiple risk aversions}},
    year = {2022},
    booktitle = {Proceedings of the ACM International Conference on AI in Finance},
    pages = {361--368},
    doi = {10.1145/3533271.3561731},
}

@article{deep_rl_hedging2,
    author = {Marzban, Saeed and Delage, Erick and Yu-Meng Li, Jonathan},
    title = {Deep reinforcement learning for option pricing and hedging under dynamic expectile risk measures},
    journal = {Quantitative Finance},
    volume = {23},
    number = {10},
    pages = {1411--1430},
    year = {2023},
    doi = {10.1080/14697688.2023.2244531},
}

@inproceedings{deep_rl_hedging3,
    author = {Malekzadeh, Parvin and Poulos, Zissis and Chen, Jacky and Wang, Zeyu and Plataniotis, Konstantinos N},
    title = {{EX-DRL: Hedging against heavy losses with extreme distributional reinforcement learning}},
    year = {2024},
    doi = {10.1145/3677052.3698668},
    booktitle = {Proceedings of the ACM International Conference on AI in Finance},
    pages = {370--378},
    numpages = {9},
}

@inproceedings{adversarial_hedging,
    author = {Hirano, Masanori and Minami, Kentaro and Imajo, Kentaro},
    title = {{Adversarial deep hedging: Learning to hedge without price process modeling}},
    year = {2023},
    booktitle = {Proceedings of the ACM International Conference on AI in Finance},
    pages = {19--26},
    doi = {10.1145/3604237.3626846},
}

@inproceedings{sig_former,
    author = {Tong, Anh and Nguyen-Tang, Thanh and Lee, Dongeun and Tran, Toan M and Choi, Jaesik},
    title = {{SigFormer: Signature transformers for deep hedging}},
    year = {2023},
    booktitle = {Proceedings of the ACM International Conference on AI in Finance},
    pages = {124--132},
    doi = {10.1145/3604237.3626841},
}

@article{ntb_net,
    title   = {{No-transaction band network: A neural network architecture for efficient deep hedging}},
    author  = {Imaki, Shota and Imajo, Kentaro and Ito, Katsuya and Minami, Kentaro and Nakagawa, Kei},
    journal = {The Journal of Financial Data Science},
    volume  = {5},
    number  = {2},
    pages   = {84--99},
    year    = {2023},
    doi = {10.3905/jfds.2023.1.125},
}

@inproceedings{fast_deep_hedging,
    author = {Mueller, Konrad and Akkari, Amira and Gonon, Lukas and Wood, Ben},
    title = {Fast deep hedging with second-order optimization},
    year = {2024},
    isbn = {9798400710810},
    doi = {10.1145/3677052.3698604},
    booktitle = {Proceedings of the ACM International Conference on AI in Finance},
    pages = {319--327},
    numpages = {9},
}

@inproceedings{proto_hedging,
    author = {Faloughi, Lisa and Guo, Ce and Luk, Wayne},
    title = {{ProtoHedge: Interpretable hedging with market prototypes}},
    year = {2025},
    booktitle = {Proceedings of the ACM International Conference on AI in Finance},
    pages = {202--210},
    doi = {10.1145/3768292.3770347},
}

@inproceedings{get_real,
    author = {Vyetrenko, Svitlana and Byrd, David and Petosa, Nick and Mahfouz, Mahmoud and Dervovic, Danial and Veloso, Manuela and Balch, Tucker},
    title = {{Get real: Realism metrics for robust limit order book market simulations}},
    year = {2021},
    booktitle = {Proceedings of the ACM International Conference on AI in Finance},
    articleno = {2},
    doi = {10.1145/3383455.3422561},
}

@InProceedings{realism_based_evaluation,
    author = {Hoang, Duc Quang and Dang, Tran Khanh and Nguyen, Van Tam},
    title = {{Evaluation of synthetic data generation for time series: Privacy, security, and applications in finance}},
    booktitle = {Proceedings of the Future Data and Security Engineering},
    year = {2026},
    doi = {10.1007/978-981-95-4724-1_9},
}

@article{stylized_facts,
    author = {Cont, Rama},
    title = {{Empirical properties of asset returns: Stylized facts and statistical issues}},
    journal = {Quantitative Finance},
    year = {2001},
    volume = {1},
    number = {2},
    pages = {223--236},
    doi = {10.1080/713665670},
}

@article{empirical_deep_hedging,
    author = {Mikkil{\"a}, Oskari and Kanniainen, Juho},
    title = {Empirical deep hedging},
    journal = {Quantitative Finance},
    volume = {23},
    number = {1},
    pages = {111--122},
    year = {2023},
    doi = {10.1080/14697688.2022.2136037},
}

@Article{deep_hedging_w_jump,
    author = {Horvath, Blanka and Teichmann, Josef and \v{Z}uri\v{c}, \v{Z}an},
    title = {Deep hedging under rough volatility},
    journal = {Risks},
    volume = {9},
    year = {2021},
    number = {7},
    article-number = {138},
    doi = {10.3390/risks9070138},
}

@inproceedings{deeper_hedging,
    author = {Gao, Kang and Weston, Stephen and Vytelingum, Perukrishnen and Stillman, Namid and Luk, Wayne and Guo, Ce},
    title = {{Deeper hedging: A new agent-based model for effective deep hedging}},
    year = {2023},
    booktitle = {Proceedings of the ACM International Conference on AI in Finance},
    pages = {270--278},
    doi = {10.1145/3604237.3626913},
}

@Article{deep_hedging_w_gan,
    author = {Boursin, Nicolas and Remlinger, Carl and Mikael, Joseph},
    title = {Deep generators on commodity markets application to deep hedging},
    journal = {Risks},
    volume = {11},
    year = {2023},
    number = {1},
    article-number = {7},
    doi = {10.3390/risks11010007}
}

@inproceedings{cofindiff,
    title     = {{CoFinDiff: Controllable financial diffusion model for time series generation}},
    author    = {Tanaka, Yuki and Hashimoto, Ryuji and Takayanagi, Takehiro and Piao, Zhe and Murayama, Yuri and Izumi, Kiyoshi},
    booktitle = {Proceedings of the International Joint Conference on
               Artificial Intelligence},
    pages     = {9357--9365},
    year      = {2025},
    doi = {10.24963/ijcai.2025/1040},
}

@inproceedings{difflob,
  title     = {{DiffLOB}: {D}iffusion models for counterfactual generation in limit order books},
  author    = {Wang, Zhuohan and Ventre, Carmine},
  booktitle = {Proceedings of the International Joint Conference on
               Artificial Intelligence},
  pages     = {6564--6572},
  year      = {2026},
  doi       = {10.24963/ijcai.2026/730},
}

@article{heston,
    author = {Heston, Steven L.},
    title = {A closed-form solution for options with stochastic volatility with applications to bond and currency options},
    journal = {The Review of Financial Studies},
    volume = {6},
    number = {2},
    pages = {327--343},
    year = {1993},
    doi = {10.1093/rfs/6.2.327},
}

@article{merton,
    title = {Option pricing when underlying stock returns are discontinuous},
    journal = {Journal of Financial Economics},
    volume = {3},
    number = {1},
    pages = {125--144},
    year = {1976},
    author = {Robert C. Merton},
    doi = {10.1016/0304-405X(76)90022-2}
}

@article{andersen_qe_m,
  author    = {Leif B. G. Andersen},
  title     = {{Simple and efficient simulation of the Heston stochastic volatility model}},
  journal   = {The Journal of Computational Finance},
  volume    = {11},
  number    = {3},
  pages     = {1--42},
  year      = {2008},
  doi = {10.21314/jcf.2008.189},
}

@article{stylized_facts_distance,
    author  = {Gao, Kang and Vytelingum, Perukrishnen and Luk, Wayne and Weston, Stephen and Guo, Ce},
    title   = {{Understanding intra-day price formation process by agent-based financial market simulation: Calibrating the extended Chiarella model}},
    journal = {Wilmott},
    volume  = {2022},
    number  = {119},
    pages   = {22--38},
    year    = {2022},
    doi = {10.54946/wilm.11014},
}

@article{hill,
    author = {Bruce M. Hill},
    title = {A simple general approach to inference about the tail of a distribution},
    journal = {The Annals of Statistics},
    volume = {3},
    number = {5},
    pages = {1163--1174},
    year = {1975},
    doi = {10.1214/aos/1176343247},
}

\newpage\appendix

\section{Hedging Setup and Risk Measures}\label{App:hedging_setup}
This appendix supplements Section~\ref{Sec:common_hedging_setup} with the payoff and risk conventions used in both experiments. The numerical settings of the controlled experiment and real-world experiment are given separately in Appendices~\ref{App:controlled_details} and~\ref{App:real_world_details}.

\paragraph{Payoffs and trading costs.} Let $\Phi(\bm\omega)$ denote the option payoff: $\Phi(\bm\omega)=(S_T-K)^+$ for the European call and $\Phi(\bm\omega)=(\max_{0\leq t\leq T}S_t-K)^+$ for the Lookback call. Here $K$ is the strike and $x^+=\max(x,0)$. We take the short option position and hedge by trading the underlying asset. Terminal hedging loss $C(\bm\omega,\bm h)$ is the option payoff plus transaction costs, less trading gains and the premium. Let $\kappa$ denote the proportional transaction cost rate. A change from $h_{t-1}$ to $h_t$ incurs a cost of $\kappa S_t|h_t-h_{t-1}|$. The controlled experiment starts with zero holdings and zero premium, charges the initial purchase, and values the final position without a liquidation cost. Its loss is therefore
\begin{align}
\begin{split}
C(\bm\omega,\bm h)&=\Phi(\bm\omega)-\sum_{t=0}^{T-1}h_t(S_{t+1}-S_t)\\
&~+\kappa\sum_{t=0}^{T-1}S_t|h_t-h_{t-1}|,\quad h_{-1}=0.
\end{split}
\end{align}
All comparisons within a task use the same payoff and cost convention.

\paragraph{Risk measures.}
ERM denotes the entropic risk measure. For risk aversion $\lambda>0$,
\begin{align}
\rho_\mu^{\mathrm{ERM}}(\bm h;\lambda)=\frac{1}{\lambda}\log\mathbb E_{\bm\omega\sim\mu}\!\left[\exp\{\lambda C(\bm\omega,\bm h)\}\right].
\end{align}
CVaR denotes Conditional Value-at-Risk. For confidence level $\alpha\in(0,1)$,
\begin{align}
\begin{split}
\rho_\mu^{\mathrm{CVaR}}(\bm h;\alpha)=\inf_{z\in\mathbb R}\Bigl\{z
&+\frac{1}{1-\alpha}\mathbb E_{\bm\omega\sim\mu}\\[-0.5ex]
&\bigl[(C(\bm\omega,\bm h)-z)^+\bigr]\Bigr\}.
\end{split}
\end{align}
The real number $z$ is a loss threshold. For a continuous loss distribution, this equals the mean loss in its upper $1-\alpha$ tail. The threshold formulation also specifies CVaR when there is probability mass at the quantile. Empirical risk replaces the expectation by a sample average. The controlled experiment uses the upper $5\%$ empirical tail, with fractional weight at the quantile boundary when necessary. Losses from all evaluation chunks are pooled before computing risk.

\section{Controlled Experiment Details}\label{App:controlled_details}
This appendix supplements the controlled experiment in Sections~\ref{Sec:controlled_experiment} and~\ref{Sec:controlled_results}. Tables~\ref{Tab:controlled_settings} and~\ref{Tab:controlled_estimation} collect the fixed simulation and estimation settings. Table~\ref{Tab:controlled_task_summaries} reports the numerical results behind Figure~\ref{Fig:compatibility_degradation_scatter}.

\begin{table}[t]
\centering
\caption{Controlled experiment: market and strategy settings. Heston parameters are prescribed for this experiment, independently of the real-world calibration.}
\label{Tab:controlled_settings}
\small
\begin{tabular}{@{}p{0.43\columnwidth}p{0.49\columnwidth}@{}}
\toprule
Setting & Value \\
\midrule
Target distribution $P$ & Heston; initial price $1$ \\
Initial / long-run variance & $0.04$ / $0.04$ \\
Variance mean reversion & $2$ \\
Volatility of variance & $0.3$ \\
Price--variance correlation & $-0.7$ \\
Synthetic distribution $Q_\theta$ & GBM; zero drift, volatility $0.20$ \\
Time step / trading steps & $1/250$ year / $20$ \\
Option types & European and Lookback calls \\
Strike prices & $0.98$, $1.00$, $1.02$ \\
Transaction cost rates & $10^{-4}$, $10^{-3}$ \\
Training risk & CVaR$(0.95)$ \\
Admissible positions & $[0,1]$ \\
Band-width network & Four hidden layers, width $32$ \\
Hidden-layer operations & Linear, LayerNorm, ReLU \\
Band half-widths & $0.25$ times sigmoid of two linear outputs \\
Delta volatility floor & $10^{-4}$ \\
\bottomrule
\end{tabular}
\end{table}

\begin{table}[t]
\centering
\caption{Controlled experiment: sample sizes and optimization. Path counts are per distribution and shared across the 12 tasks.}
\label{Tab:controlled_estimation}
\small
\begin{tabular}{@{}p{0.43\columnwidth}p{0.49\columnwidth}@{}}
\toprule
Setting & Value \\
\midrule
Training paths & $10{,}000$ \\
Validation paths & $20{,}000$ \\
Search paths & $20{,}000$ \\
Test paths & $100{,}000$ \\
Training trials & $50$ per distribution and task \\
Training epochs & $500$ \\
Optimizer / learning rate & Adam / $10^{-3}$ \\
Minibatch size & $1{,}024$ \\
Training model selection & Validation at epoch $0$ and every $50$ epochs \\
Complete searches per task & $10$ \\
Search objectives & Both signs of the risk difference \\
Initializations per sign & Two random policies and the $P/Q_\theta$ references \\
Updates per trajectory & $1{,}000$ Adam updates \\
Search selection & Maximum absolute validation gap at steps $0,100,\ldots,1000$ \\
Evaluation chunk size & $4{,}096$ paths\\
\bottomrule
\end{tabular}
\end{table}

\paragraph{Market simulation and observations.}
We use the PFHedge~\cite{pfhedge} Heston simulator for $P$ and a zero-drift GBM simulator for $Q_\theta$. Each path contains the initial price and 20 subsequent prices. The option's running maximum includes both the initial and terminal prices. The policy does not observe the latent Heston variance or the generating GBM volatility.

NTBNet receives log-moneyness and remaining time together with the previous holding, realized volatility, and log-moneyness of the running maximum. The policy's running maximum uses prices observed through the current date. At each date, realized volatility is computed by dividing the mean squared log return observed since inception by the time step and taking the square root; it is zero at inception. The volatility floor in Table~\ref{Tab:controlled_settings} applies only to the delta calculation.

\paragraph{Admissible strategy class.}
The network predicts lower and upper half-widths around the Black--Scholes delta for the specified option. Band endpoints are clipped to $[0,1]$, and the previous holding is projected onto the resulting band. The same policy construction is used in training and compatibility search.

The option and strike affect the reference delta and the interpretation of the inputs. Accordingly, the experiment fixes the architecture and position constraints, while $\mathcal H$ denotes the resulting admissible strategy class for each task. It does not establish equality of the induced price-history-to-position function classes across all tasks. The fixed distributional realism comparison follows from the unchanged $P$ and $Q_\theta$.

\paragraph{Training and numerical references.} For each distribution, one epoch visits every fixed training path in minibatches. Each trial retains the checkpoint with the lowest validation risk. The common reference $\tilde{\bm h}_P$ is selected from the 50 trial minima using $P$ validation risk. We similarly define $\tilde{\bm h}_{Q_\theta}$ using $Q_\theta$ validation risk. References are specific to a task and are shared only across that task's trial comparisons. Test data are excluded from selection. Minibatch CVaR training does not guarantee minimization of the full-sample empirical objective.

\paragraph{Compatibility search.} Each complete search runs both signs of the risk difference from the initializations in Table~\ref{Tab:controlled_estimation}. Search updates draw independent minibatches with replacement from the fixed $P$ and $Q_\theta$ search datasets. Checkpoints are compared by their absolute validation risk difference; the two numerical references are also candidates. The selected strategy is frozen and evaluated on the two test datasets. Thus, test paths are used neither for gradient updates nor for selecting a search checkpoint. The 10 complete searches retain the same datasets and reference strategies.

\begin{table*}[t]
\centering
\caption{Controlled experiment: all 12 task summaries. Gap and risk differences are in units of $10^{-3}$. Compatibility entries summarize 10 searches; other entries summarize 50 training trials. The final column estimates learning error under $Q_\theta$ as the test-risk difference between each trained strategy and the common reference $\tilde{\bm h}_{Q_\theta}$, selected by the lowest validation risk across 50 training trials.}
\label{Tab:controlled_task_summaries}
\small
\setlength{\tabcolsep}{5pt}
\begin{tabular}{@{}lccrrrr@{}}
\toprule
Option & Strike & Fee & Median $\hat\Gamma$ & Range of $\hat\Gamma$
& Median $\hat D$ & Reference diagnostic: median [min, max] \\
\midrule
European & 0.98 & $10^{-4}$ & 4.0221 & [3.7240, 4.8529] & 0.1759 & 0.0523 [0.0000, 0.1196] \\
European & 0.98 & $10^{-3}$ & 4.0592 & [3.6286, 4.8369] & 0.2013 & 0.0379 [$-0.0005$, 0.0958] \\
European & 1.00 & $10^{-4}$ & 5.1001 & [4.9281, 5.2691] & 0.3763 & $-0.0047$ [$-0.0441$, 0.0864] \\
European & 1.00 & $10^{-3}$ & 4.8482 & [4.4826, 5.0671] & 0.4938 & 0.0477 [$-0.0061$, 0.1220] \\
European & 1.02 & $10^{-4}$ & 4.1660 & [3.3331, 4.3899] & 0.0800 & 0.0446 [0.0000, 0.1000] \\
European & 1.02 & $10^{-3}$ & 4.1467 & [3.4022, 4.1662] & 0.1484 & 0.0383 [0.0000, 0.1139] \\
Lookback & 0.98 & $10^{-4}$ & 6.1530 & [5.9355, 6.4764] & 1.3493 & 0.0515 [$-0.0141$, 0.1058] \\
Lookback & 0.98 & $10^{-3}$ & 6.2002 & [5.7465, 6.4630] & 1.3691 & 0.0327 [$-0.0358$, 0.1141] \\
Lookback & 1.00 & $10^{-4}$ & 6.2176 & [5.8258, 6.5357] & 1.3615 & 0.0343 [$-0.0031$, 0.1048] \\
Lookback & 1.00 & $10^{-3}$ & 5.8908 & [5.6961, 6.4510] & 1.3749 & 0.0945 [0.0000, 0.1671] \\
Lookback & 1.02 & $10^{-4}$ & 5.2235 & [5.0562, 5.5593] & 0.6991 & 0.0584 [0.0000, 0.1440] \\
Lookback & 1.02 & $10^{-3}$ & 5.4248 & [5.0927, 5.8981] & 0.7212 & $-0.0501$ [$-0.0956$, 0.0192] \\
\bottomrule
\end{tabular}
\end{table*}

\paragraph{Reference diagnostic and interpretation.}
As an empirical estimate of learning error under $Q_\theta$, we report the test-risk difference
\begin{align}
\hat\rho_{Q_\theta}^{\mathrm{test}}(\hat{\bm h}_{Q_\theta})-\hat\rho_{Q_\theta}^{\mathrm{test}}(\tilde{\bm h}_{Q_\theta}).
\end{align}
This measures how much additional risk the trained strategy incurs relative to the selected reference $\tilde{\bm h}_{Q_\theta}$, which serves as a numerical approximation to the unavailable optimal strategy $\bm h_{Q_\theta}^*$. Its task medians range from $-5.01\times10^{-5}$ to $9.45\times10^{-5}$, whereas median $\hat D$ ranges from $8.00\times10^{-5}$ to $1.375\times10^{-3}$. These values suggest relatively modest learning error when measured against the selected reference. However, this diagnostic does not capture optimization error shared by the trained strategies and the reference. Negative differences are retained because the reference is selected on validation data and evaluated on independent test data. Neither numerical reference is guaranteed to be optimal. Similarly, $\hat\Gamma$ estimates the achieved risk gap of a selected strategy, not a certified population lower bound or the true supremum. 

\begin{table}[t]
\centering
\caption{Real-world experiment: data and task settings.}
\label{Tab:real_world_settings}
\small
\begin{tabular}{@{}p{0.38\columnwidth}p{0.56\columnwidth}@{}}
\toprule
Setting & Specification \\
\midrule
Target market & Daily N225 \\
Calibration period & 2018-01-01 to 2022-12-31 \\
Calibration observations & $1{,}218$ trading days \\
Test period & 2023-01-01 to 2026-02-28 \\
Test observations & $771$ trading days \\
Initial price / strike & $1$ / $1$ \\
Maturity & $20$ days ($20$ decision steps) \\
Option types & European and Lookback calls \\
Transaction cost rates & $10^{-4}$, $10^{-3}$ \\
Risk settings & ERM$(50)$; CVaR$(0.90)$; CVaR$(0.95)$ \\
Ranking & Within each hedger--task setting and risk measure; then mean rank \\
\bottomrule
\end{tabular}
\end{table}

\begin{table}[t]
\centering
\caption{Assets used to train the VAE in the real-world experiment, all restricted to the calibration period. Names follow the data tickers.}
\label{Tab:VAEAssets}
\small
\begin{tabular}{@{}p{0.23\columnwidth}p{0.69\columnwidth}@{}}
\toprule
Group & Tickers \\
\midrule
Equity indices & N225, GSPC, DJI, IXIC, RUT, FTSE, GDAXI, FCHI, STOXX, HSI,
KS11, NSEI, AXJO \\
Exchange-traded funds & SPY, QQQ, XLF, XLK \\
Other series & VIX, TNX, NQ=F, JPY=X, EURUSD=X, BTC-USD, ETH-USD \\
\bottomrule
\end{tabular}
\end{table}

\section{Real-World Experiment Details}\label{App:real_world_details}
This appendix supplements Section~\ref{Sec:real_world_experiment} with the data and model specifications used for the real-world experiment. These settings are distinct from the prescribed distributions and bounded NTBNet used in the controlled experiment. MLP denotes a multilayer perceptron.

\paragraph{Data and calibration.} The calibration and test periods in Table~\ref{Tab:real_world_settings} are disjoint. The classical generators are initialized by moment matching and then calibrated by L-BFGS-B to minimize the stylized facts distance~\cite{stylized_facts_distance}. GBM and Merton use zero drift. Heston paths use Andersen's quadratic-exponential discretization~\cite{andersen_qe_m}. The VAE is trained on the calibration period, augmented with the additional assets listed in Table~\ref{Tab:VAEAssets}. Test-period N225 data are used to evaluate realism and hedging performance. For out-of-sample evaluation, we construct non-overlapping 20-day price paths from the test period to avoid dependence induced by overlapping evaluation windows.

\begin{table*}[t]
\centering
\caption{Real-world experiment: model specifications.}
\label{Tab:real_world_models}
\small
\begin{tabular}{@{}p{0.18\textwidth}p{0.4\textwidth}p{0.35\textwidth}@{}}
\toprule
Model & Structure & Inputs or calibration \\
\midrule
GBM & Geometric Brownian motion with zero drift & N225 calibration-period stylized facts \\
Merton~\cite{merton} & Jump-diffusion process with zero drift & N225 calibration-period stylized facts \\
Heston~\cite{heston} & Stochastic volatility; quadratic-exponential simulation & N225 calibration-period stylized facts \\
VAE & Three-layer MLP encoder and decoder; hidden width $16$, latent dimension $16$;
ReLU hidden activations, identity decoder output & Calibration-period asset data in Table~\ref{Tab:VAEAssets} \\
\midrule
LinearHedger & Linear position rule & Log-moneyness, time-to-expiry, realized volatility over the previous 10 days, previous position \\
NTBNet & Band prediction by a four-hidden-layer MLP of width $32$;
LayerNorm and ReLU, identity network output & Black--Scholes delta and previous position \\
MLPHedger & Direct position prediction by a four-hidden-layer MLP of width $32$;
LayerNorm and ReLU, identity output & Same features as LinearHedger \\
\bottomrule
\end{tabular}
\end{table*}

\begin{table*}[t]
\centering
\caption{Definitions of the real-world experiment's realism distances. All are discrepancies from the N225 test data; smaller values indicate closer agreement. Returns are daily log returns.}
\label{Tab:realism_definitions}
\small
\begin{tabular}{@{}p{0.18\textwidth}p{0.75\textwidth}@{}}
\toprule
Distance & Definition \\
\midrule
KS & Kolmogorov--Smirnov distance between empirical return distributions. \\
MeanAbsRet & Absolute difference in mean absolute return. \\
Kurt & Absolute difference in return kurtosis. \\
Skew & Absolute difference in return skewness. \\
Hill (Avg) & Mean absolute difference of Hill tail-index estimates~\cite{hill} for positive,
negative, and absolute returns, each estimated from its largest $5\%$ of
tail magnitudes. \\
ACF (Raw) & Mean absolute difference between return autocorrelations at lags $1$--$9$. \\
ACF (Abs) & Mean absolute difference between absolute-return autocorrelations at lags $1$--$9$. \\
Hurst (Abs) & Absolute difference in estimated Hurst exponents of absolute returns. \\
\bottomrule
\end{tabular}
\end{table*}

\paragraph{Hedger training and ranking.}
Table~\ref{Tab:real_world_models} specifies the three hedger classes. For each generator--hedger--task combination, separate strategies are trained for the three risk settings in Table~\ref{Tab:real_world_settings}. A strategy is evaluated using the same risk measure and parameter as its training objective. Within each hedger--task setting, generators receive ranks from $1$ (lowest test risk) to $4$ for each risk measure. Table~\ref{Tab:generator_ranking} averages those ranks over the three risk settings. Rankings summarize relative performance; they do not quantify the size or statistical significance of risk differences.

\section{Realism Distances and Additional Visualizations} \label{App:realism_diagnostics}
This appendix supplements the realism evaluation in Sections~\ref{Sec:real_world_experiment} and~\ref{Sec:real_world_results}. Table~\ref{Tab:realism_definitions} defines the distances in Table~\ref{Tab:realism_metrics}. Figures~\ref{Fig:real_market_returns}, \ref{Fig:synthetic_returns}, and~\ref{Fig:synthetic_acf} provide the corresponding distributional and temporal diagnostics.

The N225 test data exhibit heavy tails and negative skewness, with reported kurtosis $15.943$, absolute-return Hill index $2.557$, and skewness $-0.733$. Absolute returns also exhibit positive short-lag autocorrelation, with an estimated Hurst exponent of $0.642$.

The generators capture different features. GBM performs the worst across almost all metrics. It fails to capture both heavy tails and temporal dependencies observed in real data. Merton improves the tail behavior by incorporating jump components, achieving better performance in tail-related metrics (e.g., Kurt). Heston provides a more balanced representation by introducing stochastic volatility, which enables it to capture both distributional properties and temporal dynamics to a certain extent. This is reflected in its relatively strong performance in both ACF-based and distributional metrics (e.g., KS). VAE, as a data-driven generative model, shows strength in capturing long-range dependencies, leading to improved performance in metrics related to temporal structure (e.g., Hurst exponent). Overall, these results highlight that different generators specialize in different aspects of realism such as distributional features and temporal dependence, suggesting that realism is inherently multi-dimensional and cannot be captured by a single metric. The histograms and correlograms supplement the numeric comparisons.

\section{Optimization Diagnostics for the Real-World Experiment}\label{App:learnability_diagnostics}

\begin{figure}[t]
\centering
\includegraphics[width=\columnwidth]{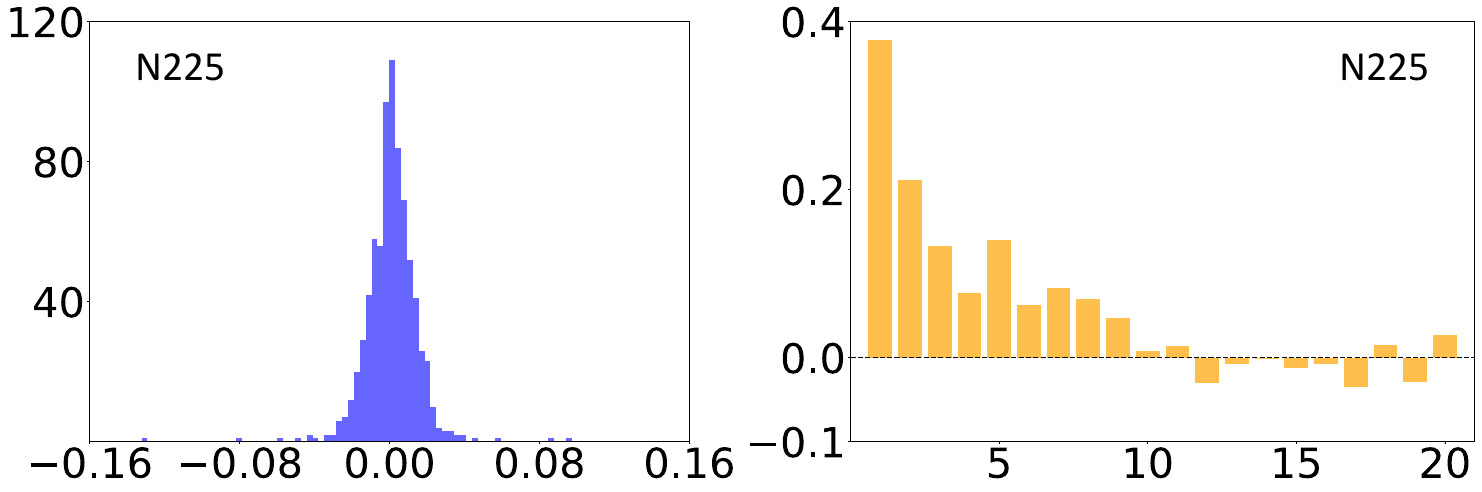}
\caption{Real-world experiment: N225 test-period daily returns
(2023-01-01 to 2026-02-28). Left: return histogram. Right: autocorrelation
function of absolute returns.}
\label{Fig:real_market_returns}
\end{figure}

\begin{figure}[t]
\centering
\includegraphics[width=\columnwidth]{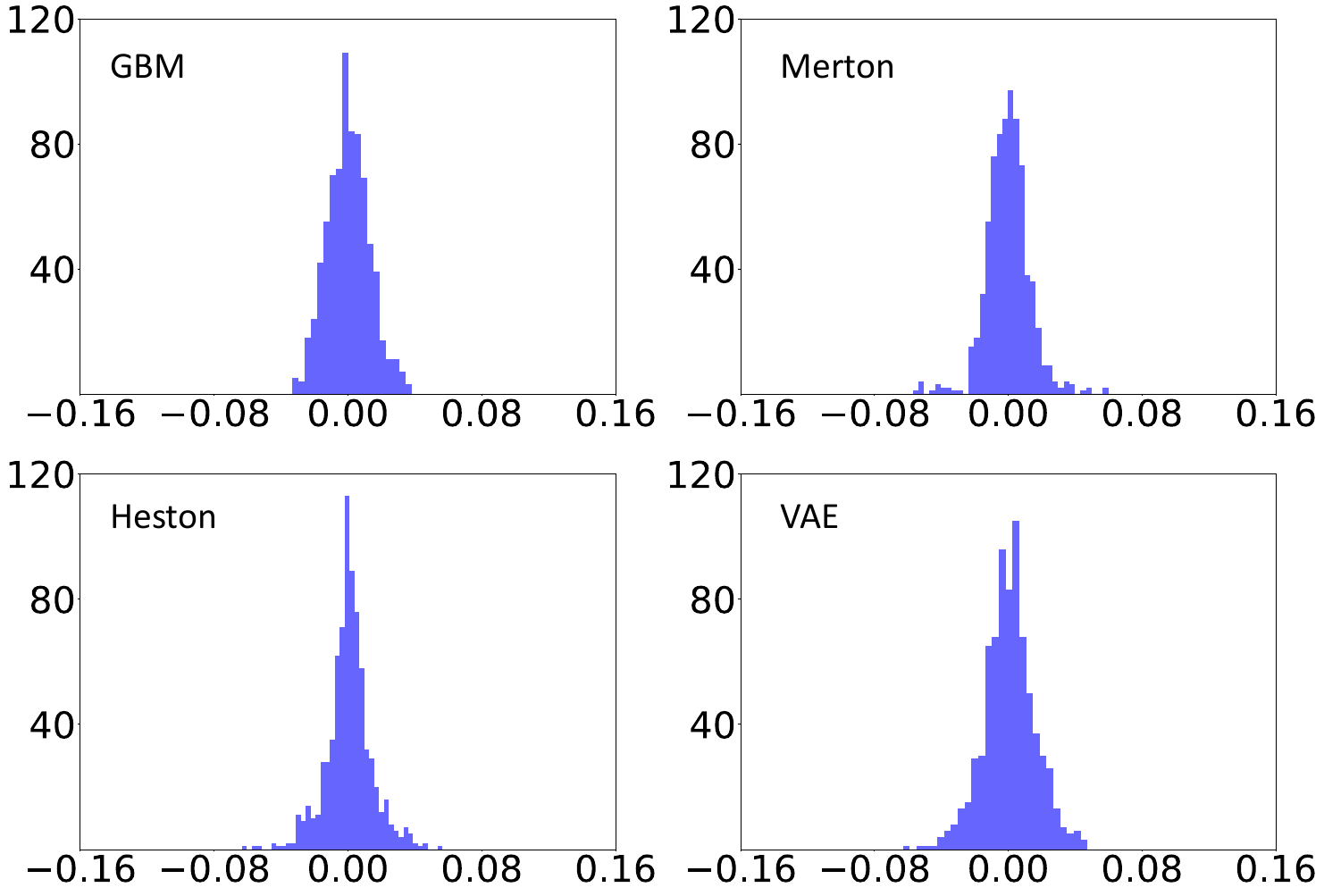}
\caption{Real-world experiment: return histograms from the four calibrated
generators.}
\label{Fig:synthetic_returns}
\end{figure}

\begin{figure}[t]
\centering
\includegraphics[width=\columnwidth]{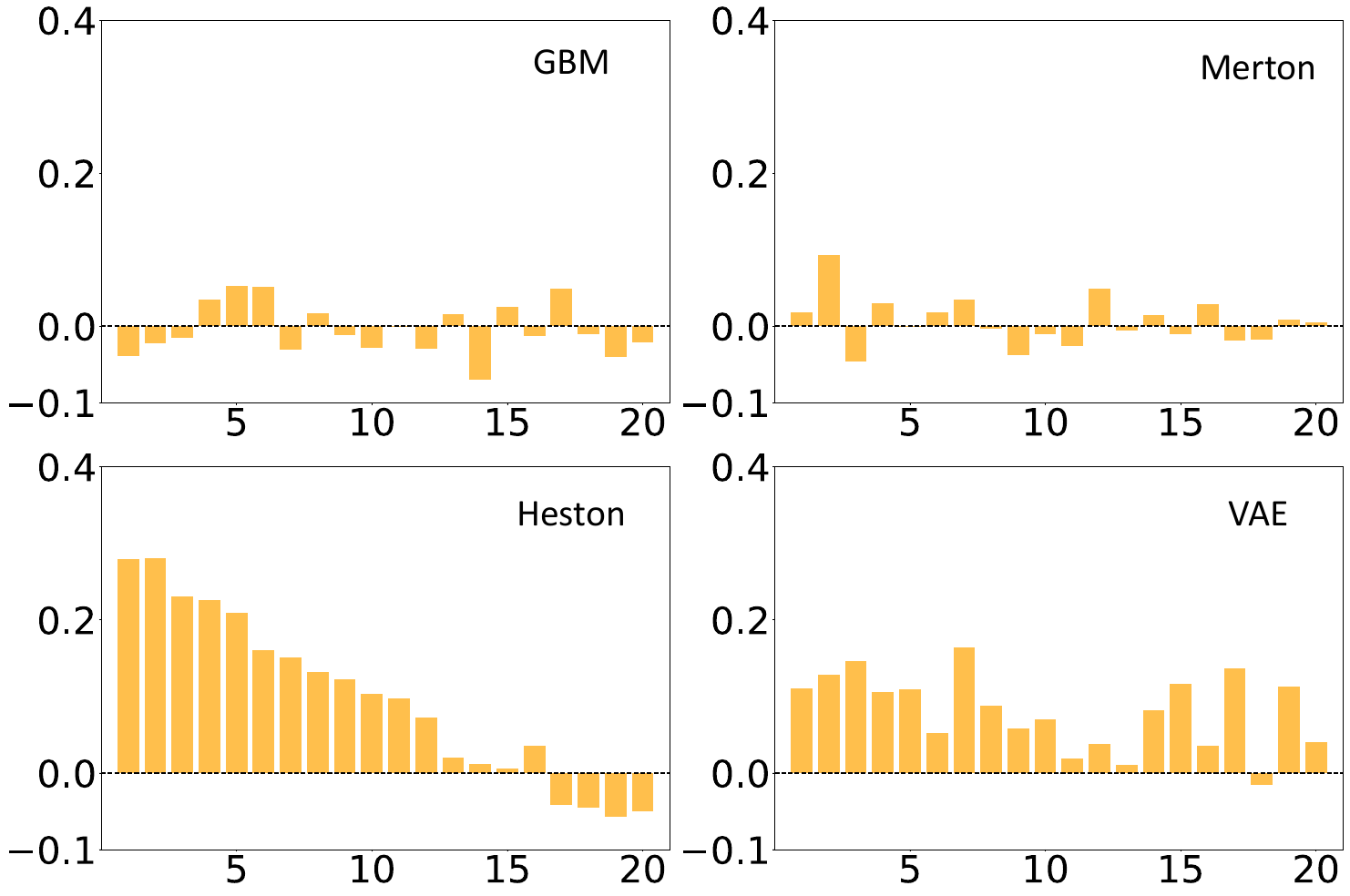}
\caption{Real-world experiment: autocorrelation functions of absolute
returns from the four calibrated generators. The visualizations extend to
lag $20$; the reported ACF distances use lags $1$--$9$.}
\label{Fig:synthetic_acf}
\end{figure}

This appendix supplements the interpretation of the real-world results in Section~\ref{Sec:real_world_results}. It reports optimization diagnostics for that experiment; these differ from the controlled experiment's numerical reference comparisons in Appendix~\ref{App:controlled_details}.

Let $L_{\mathrm{train}}(e)$ and $L_{\mathrm{valid}}(e)$ denote training and validation risk at epoch $e$. The relative generalization gap and validation variation are
\begin{align}
\begin{split}
\mathrm{GenGap}_e
&=\frac{L_{\mathrm{valid}}(e)-L_{\mathrm{train}}(e)}
{L_{\mathrm{valid}}(e)},\\
\mathrm{Stability}_e
&=\bigl|L_{\mathrm{valid}}(e)-L_{\mathrm{valid}}(e-1)\bigr|.
\end{split}
\end{align}
The diagnostics summarize epochs $e\geq100$. Table~\ref{Tab:real_world_optimization} reports their means and standard deviations. These are empirical training diagnostics, not estimates of the learning error relative to a population optimum. Their reported standard deviations are not confidence intervals for market performance.

\begin{table}[t]
\centering
\caption{Real-world experiment: optimization diagnostics, reported as mean (standard deviation). Bold marks the smallest mean within a hedger class. The scales apply to both the mean and standard deviation.}
\label{Tab:real_world_optimization}
\small
\setlength{\tabcolsep}{4pt}
\begin{tabular}{@{}llrr@{}}
\toprule
Hedger & Generator & GenGap ($10^{-2}$) & Stability ($10^{-4}$) \\
\midrule
LinearHedger & GBM & \textbf{2.847} (0.031) & \textbf{0.219} (0.164) \\
& Merton & 7.296 (0.111) & 0.379 (0.280) \\
& Heston & 3.599 (0.038) & 0.256 (0.187) \\
& VAE & 3.552 (0.030) & 0.275 (0.196) \\
\midrule
NTBNet & GBM & \textbf{7.782} (0.787) & \textbf{0.385} (0.426) \\
& Merton & 19.627 (1.339) & 0.546 (0.549) \\
& Heston & 11.354 (0.593) & 0.471 (0.370) \\
& VAE & 15.425 (1.375) & 0.433 (0.381) \\
\midrule
MLPHedger & GBM & \textbf{4.295} (0.049) & 0.309 (0.218) \\
& Merton & 4.398 (0.048) & 0.378 (0.209) \\
& Heston & 4.496 (0.054) & \textbf{0.296} (0.223) \\
& VAE & 4.579 (0.070) & 0.354 (0.278) \\
\bottomrule
\end{tabular}
\end{table}

Across hedger classes, GBM consistently exhibits the best learnability, achieving the lowest GenGap and Stability. This suggests that its simple structure makes it easier for hedgers to approximate the optimal strategy under the synthetic distribution. In contrast, the other generators show comparable performance to each other. Notably, for the MLPHedger, the performance gap between GBM and the other generators becomes less pronounced. This suggests that the higher representational capacity of the MLPHedger allows it to absorb structural differences across generators, reducing the impact of generator complexity on learnability. Overall, these results indicate that learnability is jointly determined by the complexity of the generator and the expressive power of the hedger, rather than by the generator alone.

\section{Detailed Risk Results for Generator Comparison}\label{App:generator_risk_values}
This appendix supplements Section~\ref{Sec:real_world_results} and
Table~\ref{Tab:generator_ranking} by reporting the risk values
underlying the generator rankings. Table~\ref{Tab:generator_risk_values} reports the results for each combination of hedger class, scenario generator, option type, transaction cost, and risk measure. The results confirm that the task-dependent generator rankings reflect differences in actual hedging risk rather than merely the aggregation of rankings. For example, under MLPHedger, Merton achieves lower risk than VAE for European options with low transaction costs, whereas VAE outperforms Merton when transaction costs are high. This reversal is consistent across the three risk measures. These findings further support that generator suitability depends on the hedging task, rather than statistical realism alone.

\begin{table*}[t]
\centering
\caption{
Real-world experiment: out-of-sample hedging risk for each hedger–generator–task configuration. Values are reported as mean (standard deviation) under ERM(50), CVaR(0.90), and CVaR(0.95), scaled by $10^{-3}$. Lower values indicate better hedging performance. Shaded cells indicate the best generator--hedger pair for each task, as reported in Table~\ref{Tab:best_models}. Bold values indicate the lowest risk within each hedger--task--risk configuration. \# trials: 50.
}
\label{Tab:generator_risk_values}
\begingroup
\setlength{\tabcolsep}{3.5pt}
\renewcommand{\arraystretch}{1.18}
\newcommand{\riskcell}[2]{%
  \shortstack{#1\\[-1pt]{\scriptsize (#2)}}%
}
\resizebox{\textwidth}{!}{%
\begin{tabular}{@{}ll*{12}{c}@{}}
\toprule
& & \multicolumn{3}{c}{European ($10^{-4}$)} & \multicolumn{3}{c}{European ($10^{-3}$)} & \multicolumn{3}{c}{Lookback ($10^{-4}$)} & \multicolumn{3}{c}{Lookback ($10^{-3}$)} \\
\cmidrule(lr){3-5}\cmidrule(lr){6-8}\cmidrule(lr){9-11}\cmidrule(lr){12-14}
Hedger & Generator & ERM(50) & CVaR(0.90) & CVaR(0.95) & ERM(50) & CVaR(0.90) & CVaR(0.95) & ERM(50) & CVaR(0.90) & CVaR(0.95) & ERM(50) & CVaR(0.90) & CVaR(0.95) \\
\midrule
Linear & GBM & \riskcell{\textbf{2.184}}{0.118} & \riskcell{\textbf{4.526}}{0.163}& \riskcell{4.893}{0.207} & \bestriskcell{3.172}{0.174} & \bestriskcell{6.081}{0.246} & \bestriskcell{\textbf{6.234}}{0.298} & \riskcell{\textbf{4.126}}{0.142} & \riskcell{\textbf{6.883}}{0.181} & \riskcell{\textbf{7.127}}{0.223} & \riskcell{5.638}{0.216} & \riskcell{8.526}{0.287} & \riskcell{9.043}{0.341} \\
& Merton & \riskcell{2.763}{0.192} & \riskcell{6.482}{0.231} & \riskcell{6.614}{0.267} & \riskcell{4.583}{0.252} & \riskcell{8.317}{0.326} & \riskcell{8.842}{0.379} & \riskcell{5.617}{0.247} & \riskcell{8.862}{0.319} & \riskcell{9.394}{0.358} & \riskcell{6.917}{0.308} & \riskcell{10.436}{0.384} & \riskcell{10.982}{0.426} \\
& Heston & \riskcell{2.296}{0.126} & \riskcell{4.713}{0.109} & \riskcell{\textbf{4.781}}{0.137} & \riskcell{3.896}{0.196} & \riskcell{7.064}{0.278} & \riskcell{7.518}{0.312} & \riskcell{4.482}{0.168} & \riskcell{6.926}{0.214} & \riskcell{7.372}{0.259} & \riskcell{6.146}{0.263} & \riskcell{9.271}{0.336} & \riskcell{9.794}{0.392} \\
& VAE & \riskcell{2.918}{0.437} & \riskcell{5.764}{0.582} & \riskcell{6.614}{0.693} & \riskcell{\textbf{3.158}}{0.284} & \riskcell{\textbf{5.937}}{0.417} & \riskcell{6.918}{0.526} & \riskcell{4.974}{0.326} & \riskcell{7.718}{0.472} & \riskcell{8.263}{0.561} & \riskcell{\textbf{5.217}}{0.294} & \riskcell{\textbf{7.863}}{0.438} & \riskcell{\textbf{8.381}}{0.517} \\
\midrule
NTBNet & GBM & \riskcell{2.387}{0.157} & \riskcell{5.186}{0.224} & \riskcell{5.694}{0.268} & \riskcell{4.826}{0.341} & \riskcell{8.462}{0.437} & \riskcell{9.017}{0.492} & \riskcell{5.183}{0.286} & \riskcell{8.326}{0.374} & \riskcell{8.894}{0.419} & \bestriskcell{\textbf{4.646}}{0.218} & \bestriskcell{\textbf{7.417}}{0.296} & \bestriskcell{\textbf{7.782}}{0.352} \\
& Merton & \riskcell{2.816}{0.183} & \riskcell{5.937}{0.251} & \riskcell{6.482}{0.293} & \riskcell{4.174}{0.263} & \riskcell{7.583}{0.354} & \riskcell{8.126}{0.398} & \riskcell{4.728}{0.227} & \riskcell{7.584}{0.301} & \riskcell{8.142}{0.347} & \riskcell{5.231}{0.249} & \riskcell{8.184}{0.326} & \riskcell{8.736}{0.381} \\
& Heston & \riskcell{3.274}{0.246} & \riskcell{6.718}{0.337} & \riskcell{7.263}{0.389} & \riskcell{\textbf{3.286}}{0.197} & \riskcell{\textbf{6.142}}{0.278} & \riskcell{6.913}{0.324} & \riskcell{5.892}{0.361} & \riskcell{9.471}{0.468} & \riskcell{10.063}{0.527} & \riskcell{5.783}{0.307} & \riskcell{9.642}{0.416} & \riskcell{9.917}{0.473} \\
& VAE & \riskcell{\textbf{2.063}}{0.171} & \riskcell{\textbf{4.618}}{0.238} & \riskcell{\textbf{5.127}}{0.284} & \riskcell{3.451}{0.249} & \riskcell{6.387}{0.346} & \riskcell{\textbf{6.752}}{0.412} & \riskcell{\textbf{4.594}}{0.238} & \riskcell{\textbf{7.116}}{0.329} & \riskcell{\textbf{7.847}}{0.386} & \riskcell{6.126}{0.374} & \riskcell{9.193}{0.492} & \riskcell{9.917}{0.558} \\
\midrule
MLP & GBM & \riskcell{2.374}{0.129} & \riskcell{4.826}{0.084} & \riskcell{4.917}{0.103} & \riskcell{4.817}{0.428} & \riskcell{8.273}{0.536} & \riskcell{8.846}{0.612} & \riskcell{5.736}{0.392} & \riskcell{9.184}{0.483} & \riskcell{9.752}{0.547} & \riskcell{6.387}{0.447} & \riskcell{10.126}{0.569} & \riskcell{10.714}{0.638} \\
& Merton & \bestriskcell{\textbf{2.148}}{0.116}
& \bestriskcell{\textbf{4.563}}{0.052}
& \bestriskcell{\textbf{4.672}}{0.068} & \riskcell{3.672}{0.213} & \riskcell{6.483}{0.287} & \riskcell{7.026}{0.332} & \riskcell{5.128}{0.271} & \riskcell{8.263}{0.354} & \riskcell{8.817}{0.403} & \riskcell{5.792}{0.316} & \riskcell{9.184}{0.427} & \riskcell{9.763}{0.486} \\
& Heston & \riskcell{2.816}{0.284} & \riskcell{5.742}{0.391} & \riskcell{6.317}{0.452} & \riskcell{4.226}{0.297} & \riskcell{6.988}{0.382} & \riskcell{7.194}{0.439} & \bestriskcell{\textbf{4.135}}{0.184} & \bestriskcell{\textbf{6.994}}{0.256} & \bestriskcell{\textbf{7.005}}{0.301} & \riskcell{5.263}{0.267} & \riskcell{8.316}{0.348} & \riskcell{8.892}{0.397} \\
& VAE & \riskcell{3.172}{0.573} & \riskcell{6.183}{0.684} & \riskcell{6.794}{0.792} & \riskcell{\textbf{3.447}}{0.276} & \riskcell{\textbf{5.926}}{0.391} & \riskcell{\textbf{6.372}}{0.458} & \riskcell{4.627}{0.294} & \riskcell{7.418}{0.412} & \riskcell{7.973}{0.479} & \riskcell{\textbf{4.718}}{0.281} & \riskcell{\textbf{7.482}}{0.394} & \riskcell{\textbf{8.036}}{0.463} \\
\bottomrule
\end{tabular}%
}
\endgroup
\end{table*}

\end{document}